\newif\ifarxiv
\arxivtrue

\ifarxiv{}
    \documentclass[a4paper,USenglish,cleveref, autoref,numberwithinsect]{lipics-v2019}
\else{}
    \RequirePackage{fix-cm}
    \RequirePackage{amsmath}
    \RequirePackage{amsfonts}
    \RequirePackage{amssymb}
    \documentclass[smallextended,numbook,pagebackref]{svjour3}       %
    \usepackage{amssymb}
\fi{}

\usepackage[utf8]{inputenc}
\usepackage[T1]{fontenc} 
\usepackage{graphicx,dsfont}
\usepackage{color,soul}
\usepackage{hyperref}
\hypersetup{citecolor=green!60!black,linkcolor=red!60!black,colorlinks}

\usepackage{mathtools}

\usepackage{bbm}
\usepackage{etex}

 \usepackage[disable]{todonotes}
\usepackage[matrix,arrow,ps]{xy}

\usepackage{booktabs}
\usepackage{tabularx}
\usepackage{paralist}

\usepackage[numbers,sort&compress]{natbib}

\usepackage{tikz}
\usetikzlibrary{decorations,arrows,petri,topaths,backgrounds,shapes,positioning,fit,calc,decorations.pathreplacing,patterns,intersections,decorations.pathmorphing,matrix}

\makeatletter
\newcommand{\gettikzxy}[3]{%
  \tikz@scan@one@point\pgfutil@firstofone#1\relax
  \edef#2{\the\pgf@x}%
  \edef#3{\the\pgf@y}%
}
\makeatother

\usepackage{etoolbox}

\ifarxiv{}
    \newtheorem{observation}[theorem]{Observation}
    \newtheorem{rrule}{Reduction Rule}
    \theoremstyle{definition}
    \newtheorem{construction}{Construction}
\else{}
    \makeatletter
    \let\cl@chapter\undefined
    \makeatletter
    \usepackage{cleveref}
    \newtheorem{rrule}{Reduction Rule}

    \spnewtheorem{observation}{Observation}[section]{\bfseries}{\itshape}

    \spnewtheorem{construction}{Construction}[subsection]{\bfseries}{\normalfont}

    \smartqed
\fi{}

\crefname{table}{Table}{Tables}
\crefname{figure}{Figure}{Figures}
\crefname{theorem}{Theorem}{Theorems}
\Crefname{theorem}{Thm.}{Thms.}
\crefname{definition}{Definition}{Definitions}
\crefname{corollary}{Corollary}{Corollaries}
\Crefname{corollary}{Cor.}{Cors.}
\crefname{observation}{Observation}{Observations}
\Crefname{observation}{Obs.}{Obs.}
\crefname{lemma}{Lemma}{Lemmas}
\crefname{example}{Example}{Examples}
\crefname{reduction}{Reduction}{Reductions}
\crefname{rrule}{Reduction Rule}{Reduction Rules}
\crefname{construction}{Construction}{Constructions}
\crefname{remark}{Remark}{Remarks}
\crefname{subsection}{Section}{Sections}
\crefname{section}{Section}{Sections}
\crefname{proposition}{Proposition}{Propositions}
\Crefname{proposition}{Prop.}{Props.}
\crefname{algorithm}{Algorithm}{Algorithms}

\newcommand{\problemdef}[3]{
	\begin{center}
		\begin{minipage}{0.96\textwidth}%
			\noindent
			\textsc{#1}
			\begin{compactdesc}
			 \item[\textbf{Input:}]  #2
			 \item[\textbf{Question:}]  #3
			\end{compactdesc}
		\end{minipage}
	\end{center}
}

\newcommand{\mymath}[1]{\ensuremath{\mathbb{#1}}}
\newcommand{\N}{\mymath{N}}

\newcommand{\wone}{{\mathrm{W[1]}}}

\DeclareMathOperator{\poly}{poly}

\newcommand{\coNP}{\ensuremath{\mathrm{coNP}}}
\newcommand{\NP}{\ensuremath{\mathrm{NP}}}
\newcommand{\XP}{\ensuremath{\mathrm{XP}}}
\newcommand{\W}[1]{\ensuremath{\mathrm{W}[#1]}}
\newcommand{\FPT}{\ensuremath{\mathrm{FPT}}}
\newcommand{\classP}{\ensuremath{\mathrm{P}}}
\newcommand{\unlessPK}{\ensuremath{\coNP\subseteq \NP/\poly}}

\renewcommand{\cal}[1]{\mathcal{#1}}

\newcommand{\calS}{\cal{S}}

\newcommand{\prob}[1]{\textsc{#1}}
\newcommand{\vc}{\prob{Vertex Cover}}
\newcommand{\VC}{\vc{}}

\newcommand{\msvc}{\prob{Multistage Vertex Cover}}
\newcommand{\MSVC}{\msvc{}}
\newcommand{\TG}{\cal{G}}
\newcommand{\TE}{\cal{E}}
\newcommand{\UG}{\ensuremath{G_{\downarrow}}}

\newcommand{\sydi}[2]{\ensuremath{#1\triangle #2}}
\newcommand{\sydic}[2]{\ensuremath{|\sydi{#1}{#2}|}}
\newcommand{\sdt}[2]{#1\diamond#2}
\newcommand{\yes}{\texttt{yes}}
\newcommand{\no}{\texttt{no}}

\newcommand{\RD}{$(\Rightarrow)$}
\newcommand{\LD}{$(\Leftarrow)$}

\newcommand{\tref}[1]{{\scriptsize (\Cref{#1})}}
\newcommand{\ttref}[2]{{\scriptsize (\Cref{#1},~\Cref{#2})}}
\newcommand{\lqed}{\ifarxiv{}\else{}\qed\fi{}}
\newcommand{\cqed}{\hfill$\blacklozenge$}
\newcommand{\etal}{{et~al.}}

\newcommand{\smooth}{smooth}
\newcommand{\solset}{\ensuremath{\Psi}}

\newcommand{\smod}[1]{\ensuremath{\widehat{#1}}}

\newcommand{\mytitle}{Multistage Vertex Cover}

\newcommand{\ceq}{\ensuremath{\coloneqq}}

\usepackage{xargs}
\newcommandx{\set}[2][1=1]{\ensuremath{\{#1,\ldots,#2\}}}
\newcommandx{\tlog}[3][1=,3=]{\log_{#1}^{#3}(#2)}
\newcommandx{\dist}[2][1=]{\ensuremath{\operatorname{dist}_{#1}(#2)}}
\newcommandx{\ith}[2][1=th]{#2\nobreakdash-#1}

\ifarxiv{}

	\hideLIPIcs	
	\nolinenumbers

  \title{\mytitle%
  }

  \author{Till Fluschnik}
  {Technische Universität Berlin, Algorithmics and Computational Complexity, Germany}
  {till.fluschnik@tu-berlin.de}
  {https://orcid.org/0000-0003-2203-4386}
  {Supported by the DFG, project TORE (NI 369/18).}
  
  \author{Rolf Niedermeier}
  {Technische Universität Berlin, Algorithmics and Computational Complexity, Germany}
  {rolf.niedermeier@tu-berlin.de}
  {https://orcid.org/0000-0003-1703-1236}
  {}

  \author{Valentin Rohm}
  {Technische Universität Berlin, Algorithmics and Computational Complexity, Germany}
  {valentinl.rohm@campus.tu-berlin.de}
  {}
  {}
  
  \author{Philipp Zschoche}
  {Technische Universität Berlin, Algorithmics and Computational Complexity, Germany}
  {zschoche@tu-berlin.de}
  {https://orcid.org/0000-0001-9846-0600}
  {}
  
  \authorrunning{T.~Fluschnik, R.~Niedermeier, V.~Rohm, P.~Zschoche}%

  \Copyright{T.~Fluschnik, R.~Niedermeier, V.~Rohm, P.~Zschoche}%



  \EventEditors{John Q. Open and Joan R. Access}
  \EventNoEds{2}
  \EventLongTitle{42nd Conference on Very Important Topics (CVIT 2016)}
  \EventShortTitle{CVIT 2016}
  \EventAcronym{CVIT}
  \EventYear{2016}
  \EventDate{December 24--27, 2016}
  \EventLocation{Little Whinging, United Kingdom}
  \EventLogo{}
  \SeriesVolume{42}
  \ArticleNo{23}

\else{}

    \title{\mytitle\thanks{An extended abstract of this paper appeared in 
    \emph{Proceedings of the 14th International Symposium on
    Parameterized and Exact Computation (IPEC’19)}, LIPIcs, Schloss
    Dagstuhl - Leibniz-Zentrum f\"ur Informatik, 2019, 148(14):1--14~\cite{FluschnikNRZ19}.
    T.\ Fluschnik acknowledges support by the DFG, project TORE (NI/369-18).}
    }

    \author{Till~Fluschnik, Rolf~Niedermeier, Valentin~Rohm, and Philipp~Zschoche}

    \institute{
    Till~Fluschnik \at
    Technische Universität Berlin, Algorithmics and Computational Complexity, Germany\\
    \email{till.fluschnik@tu-berlin.de}           %
    \and
    Rolf~Niedermeier \at
    Technische Universität Berlin, Algorithmics and Computational Complexity, Germany\\
    \email{rolf.niedermeier@tu-berlin.de}           %
    \and
    Valentin~Rohm \at
    Technische Universität Berlin, Algorithmics and Computational Complexity, Germany\\
    \email{valentinl.rohm@campus.tu-berlin.de}           %
    \and
    Philipp~Zschoche \at
    Technische Universität Berlin, Algorithmics and Computational Complexity, Germany\\
    \email{zschoche@tu-berlin.de}           %
    }

    \date{Received: date / Accepted: date}

\fi{}

\begin{document}
\sloppy
\allowdisplaybreaks

\maketitle

\begin{abstract}  
Covering all edges of a graph by a small number of vertices, 
this is the NP-complete \prob{Vertex Cover} problem.
It is among the most
fundamental graph-algorithmic problems. 
Following a recent trend in studying temporal graphs 
(a sequence of graphs, 
so-called layers,
over the same vertex set but, 
over time,
changing edge sets), 
we initiate the study of 
\prob{Multistage Vertex Cover}. 
Herein, given a temporal graph, 
the goal is to 
find for each layer of the temporal graph a small vertex 
cover \emph{and} to guarantee that two vertex cover sets of 
every two consecutive layers differ not too much (specified by a given parameter).
We show that, 
different from classic \prob{Vertex Cover} and
some other dynamic or temporal variants of it, 
\prob{Multistage Vertex Cover} is computationally hard even in fairly restricted settings.
On the positive side, however, 
we also spot several fixed-parameter tractability 
results based on some of the most natural parameterizations.
\end{abstract}

\section{Introduction}\label{sec:intro}

\vc{} %
asks, given an undirected graph~$G$ and an integer~$k\geq 0$, whether at most~$k$ vertices can be deleted from~$G$ such that the remaining graph contains no edge.
\VC{} is NP-complete and it is a formative problem of algorithmics and 
combinatorial optimization.
We study a \emph{time-dependent}, ``\emph{multistage}''  version, 
namely a variant of \VC{} on temporal graphs.
A \emph{temporal graph}~$\TG$ is a tuple~$(V,\TE,\tau)$ consisting of a set~$V$ of vertices, a discrete time-horizon~$\tau$, and a set of temporal edges~$\TE\subseteq \binom{V}{2}\times \set{\tau}$.
Equivalently, a temporal graph~$\TG$ can be seen as a vector~$(G_1,\ldots,G_\tau)$ of static graphs (\emph{layers}), where each graph is defined over the same vertex set~$V$.
Then, our specific goal is to find a small vertex cover~$S_i$ for each layer~$G_i$ 
such that 
the size of the symmetric difference~$\sydi{S_i}{S_{i+1}}=(S_i\setminus S_{i+1})\cup(S_{i+1}\setminus S_i)$ of the vertex covers~$S_i$ and~$S_{i+1}$ 
of every two consecutive layers~$G_i$ and~$G_{i+1}$ is small.
Formally, we thus introduce and study the following problem (see~\cref{fig:introex} for an illustrative example).

\problemdef{\msvc}
{A temporal graph~$\TG = (V,\TE,\tau)$ and two integers~$k\in\N,\ell\in\N_0$.}
{Is there a sequence~$\cal{S}=(S_1,\ldots,S_\tau)$ such that
\begin{compactenum}[(i)]
 \item for all~$i\in\{1,\ldots,\tau\}$, it holds true that $S_i\subseteq V$ is a
	 size-at-most-$k$ vertex cover for layer~$G_i$, and
 \item for all~$i\in\{1,\ldots,\tau-1\}$, it holds true that $\sydic{S_i}{S_{i+1}}\leq \ell$?
\end{compactenum}
}

Throughout this paper we assume that $0 < k < |V|$ 
because otherwise we have a trivial instance.
In our model, we follow the recently proposed \emph{multistage}~\cite{BampisELP18,GuptaTW14,BampisET19,EMS14,BEK19,HHKNRS19,FNSZ20,CTW20} view on classical optimization problems on temporal graphs.

In general, the motivation behind a multistage variant of a classical problem 
such as \prob{Vertex cover} is that the 
environment changes over time (here reflected by the changing edge 
sets in the temporal graph) 
and a corresponding adaptation of the current solution comes with a cost.
In this spirit, the parameter~$\ell$ in the definition of 
\MSVC{} allows to model that only moderate changes concerning the 
solution vertex set may be wanted when moving from one layer to the subsequent 
one. Indeed, in this sense $\ell$~can be interpreted as 
a parameter measuring the degree of (non-)conservation~\cite{HN13,AEFRS15}.

It is immediate that \MSVC{} is \NP-hard as it generalizes \vc{} ($\tau=1$).
We will study its parameterized complexity regarding the problem-specific parameters~$k$, 
$\tau$,  
$\ell$, 
and some of their combinations, 
as well as restrictions to temporal graph classes~\cite{CasteigtsFQS12,FluschnikMNZ18}.

\begin{figure}[h]
    \centering
    \begin{tikzpicture}
        \def\xr{0.88}
        \def\yr{1}
        \tikzstyle{xnode}=[circle,fill,scale=2/3,draw];
        \tikzstyle{hnode}=[fill=green,draw=none,circle,scale=1.3,draw,opacity=0.125];
        \tikzstyle{xedge}=[thick,-];
        \newcommand{\tikzlayer}[1]{
            \draw[rounded corners, gray, very thin] (0.1*\xr,0.1*\yr) rectangle (3.4*\xr,1.9*\yr);
            \node at (-0.3*\xr,1.25*\yr)[label=90:{#1}]{};
            \node (v1) at (1*\xr,1.5*\yr)[xnode,label=180:{$v_1$}]{};
            \node (v2) at (2.5*\xr,1.5*\yr)[xnode,label=0:{$v_2$}]{};
            \node (v3) at (2.5*\xr,0.5*\yr)[xnode,label=00:{$v_3$}]{};
            \node (v4) at (1*\xr,0.5*\yr)[xnode,label=180:{$v_4$}]{};
        }

        \begin{scope}[]
        \tikzlayer{$G_1$}
        \draw[xedge] (v2) -- (v3) -- (v4) -- (v2);
        \node at (v3)[hnode]{};
        \node at (v2)[hnode]{};
        \end{scope}

        \begin{scope}[xshift=\xr*4.5cm]
        \tikzlayer{$G_2$}
        \draw[xedge] (v3) -- (v1);
        \node at (v3)[hnode]{};
        \end{scope}

        \begin{scope}[xshift=\xr*9cm]
        \tikzlayer{$G_3$}
        \draw[xedge] (v2) -- (v3) -- (v4) -- (v1) -- (v2);
        \node at (v1)[hnode]{};
        \node at (v3)[hnode]{};
        \end{scope}
    \end{tikzpicture}
    \caption{An illustrative example with temporal graph~$\TG=(G_1,G_2,G_3)$ over the vertex set~$V=\{v_1,\dots,v_4\}$. A solution~$\calS=(\{v_2,v_3\},\{v_3\},\{v_1,v_3\}$) for~$k=2$ and~$\ell=1$ is highlighted.}
    \label{fig:introex}
\end{figure}
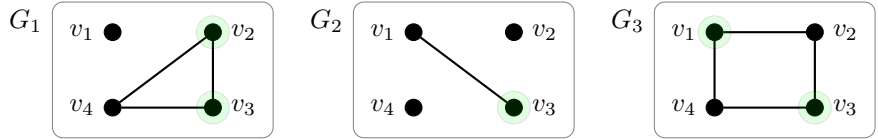

\ifarxiv{}
\subparagraph{Related Work.}
\else{}
\paragraph{Related Work.}
\fi{}
The literature on vertex covering is extremely rich, 
even when focusing on parameterized complexity studies. 
Indeed,
\vc{} %
can be seen as ``drosophila'' of parameterized algorithmics.
Thus, we only consider \VC{} studies 
closely related to our setting.
First, we mention in passing that \VC{} is studied in dynamic graphs~\cite{IwataO14,AlmanMW17} and graph stream models~\cite{ChitnisCEHMMV16}.
More importantly for our work, 
Akrida~\etal~\cite{AkridaMSZ18} studied a variant of \VC{} on temporal graphs.
Their model significantly differs from ours: 
they want an edge to be covered at least once over every time window of some given size~$\Delta$.
That is, 
they define a temporal vertex cover as a set~$S\subseteq V\times \set{\tau}$ such that,
for every time window of size~$\Delta$ and for each edge~$e=\{v,w\}$ appearing in a layer contained in the time window,
it holds that~$(v,t)\in S$ or~$(w,t) \in~S$ for some~$t$ in the time window with~$(e,t)\in\TE$.
For their model,
Akrida~\etal{} ask whether such an~$S$ of small cardinality exists.
Note that if~$\Delta>1$, 
then for some~$t\in\set{\tau}$
the set~$S_t\ceq \{v\mid (v,t)\in S\}$ is not necessarily a vertex cover of layer~$G_t$.
For~$\Delta=1$, each~$S_t$ must be a vertex cover of~$G_t$.
However, in Akrida~\etal's model the size of each~$S_t$ as well as the size of the symmetric difference between each~$S_{t}$ and~$S_{t+1}$ may strongly vary. 
They provide several hardness results and algorithms 
(mostly referring to approximation or exact algorithms, 
but not to parameterized complexity studies).

A second  related line of research, not directly referring to temporal 
graphs though, 
studies reconfiguration problems
which arise when we wish to find a step-by-step transformation
between two feasible solutions of a problem such that
all intermediate results are feasible solutions as well \cite{ito2011complexity,gopalan2009connectivity}.
Among other reconfiguration problems,
Mouawad~\etal~\cite{mouawad2017parameterized,mouawad2018vertex}
studied \prob{Vertex Cover Reconfiguration}: 
given a graph~$G$, 
two vertex covers~$S$ and~$T$ each of size at most~$k$, 
and an integer~$\tau$,
the question is whether there is a sequence~$(S=S_1,\dots,S_\tau=T)$ such that each $S_t$ is 
a vertex cover of size at most $k$.
The essential difference to our model is that from one ``sequence element'' 
to the next only one vertex may be changed 
and that the input graph does not change over time.
Indeed, there is an easy reduction of this model to ours while the opposite direction 
is unlikely to hold. 
This is substantiated by the fact that Mouawad~\etal~\cite{mouawad2017parameterized} showed that 
\prob{Vertex Cover Reconfiguration} is fixed-parameter tractable when 
parameterized by vertex cover size~$k$ while we show
W[1]-hardness for the corresponding case of \MSVC{}.

Finally, there is also a close relation to the research on 
dynamic parameterized problems~\cite{AEFRS15,KST18}. 
Krithika~\etal~\cite{KST18} studied \prob{Dynamic Vertex Cover} where
one is given two graphs on the same vertex set and a vertex cover for one 
of them together with the guarantee that the cardinality of the 
symmetric difference 
between the two edge sets is upper-bounded by a parameter~$d$. The task then is 
to find a vertex cover for the second graph that is ``close enough'' (measured by a second parameter) to the vertex cover of the first graph.
They show fixed-parameter tractability and a linear kernel with respect to~$d$.

\ifarxiv{}
\subparagraph*{Our Contributions.}
\else{}
\paragraph{Our Contributions.}
\fi{}

Our results, 
focusing on the three perhaps most natural parameters, 
are summarized in \cref{tab:results}.%
\begin{table}[t]
\centering
\caption{Overview of our results.
The column headings describe the restrictions on the input and each row corresponds to a parameter.
 p-NP-hard, PK, and NoPK abbreviate para-NP-hard, polynomial-size problem kernel, and no problem kernel of polynomial size unless~\unlessPK{}.
 }
  \setlength{\tabcolsep}{8pt}
 \begin{tabular}{@{}r|llll@{}}\toprule
            & \multicolumn{2}{l}{general layers} & tree layers & one-edge layers \\ 
            & $0\leq\ell<2k$ & $\ell\geq 2k$ & $0\leq \ell <2k$ & $1\leq \ell<2$ \\\midrule
            & \NP-hard & \NP-hard
            & \NP-hard & \NP-hard \\
            & & & {\small(\Cref{thm:npahrdcases}(i))} & {\small(\Cref{thm:npahrdcases}(ii))} \\\midrule
  $\tau$    
            & p-\NP-hard & p-\NP-hard
            &  p-\NP-hard &  FPT, PK  \\
            & \tref{thm:npahrdcases} & \tref{thm:npahrdcases}
            & \tref{thm:npahrdcases} &  \tref{thm:preproctau}  \\
  $k$       & \XP, \W{1}-h.,          & FPT$^\dagger$, NoPK  & \XP, \W{1}-h.  & \emph{open}, NoPK \\
            & \tref{thm:xpwhardness}  & \ttref{obs:turedu}{thm:preprock} & \ttref{thm:xpwhardness}{rem:whardnesstree} & \tref{thm:preprock}  \\
            
  $k+\tau$  
            & FPT, PK & FPT, PK 
            & FPT, PK & FPT, PK \\
            & \tref{thm:PKktau} & \tref{thm:PKktau}  
            & \tref{thm:PKktau} & \tref{thm:PKktau} \\
  \bottomrule
 \end{tabular}
 \label{tab:results}
\end{table}
We highlight a few specific results.
\msvc{} remains~\NP-hard even if every layer 
consists of only one edge;  
not surprisingly, 
the corresponding hardness 
reduction exploits 
an unbounded number~$\tau$ of time layers.
If one only has two layers, however, one of them 
being a tree and the other being a path, then again \msvc{}
already becomes NP-hard.
\MSVC{} parameterized by solution size~$k$ is fixed-parameter tractable if~$\ell\geq 2k$,
but becomes \W{1}-hard if~$\ell<2k$.
Considering the tractability results for \prob{Dynamic Vertex Cover} \cite{KST18}
and \prob{Vertex Cover Reconfiguration} \cite{mouawad2017parameterized}, 
this hardness is surprising,
and it is our most technical result.
Furthermore,
\MSVC{} parameterized by $k$ with $\ell\geq2k$
does not admit a problem kernel of polynomial size unless \unlessPK{}.
Finally,
for the combined parameter~$k+\tau$ 
we obtain polynomial-sized problem kernels 
(and thus fixed-parameter tractability) 
in \emph{all} cases without any further constraints.

\ifarxiv{}
\subparagraph{Outline.}
\else{}
\paragraph{Outline.}
\fi{}
In \cref{sec:prelims},
we provide some preliminaries.
For \MSVC{},
we give some first and general observations in \cref{sec:firstobs},
study the parameterized complexity regarding~$k$ in~\cref{sec:paramvc},
and 
discuss the possibilities for efficient data reduction in~\cref{sec:dataredu}.
We conclude in~\cref{sec:conclusion}.

\section{Preliminaries}
  \label{sec:prelims}

We denote by~$\N$ and~$\N_0$ the natural numbers excluding and including zero, respectively.
For two sets~$A$ and~$B$,
we denote by~$\sydi{A}{B} \ceq (A\setminus B) \cup (B \setminus A)=(A\cup B)\setminus (A\cap B)$ the symmetric difference of~$A$ and~$B$,
and by~$A\uplus B$ the disjoint union of~$A$ and~$B$.

\ifarxiv{}
\subparagraph{Temporal Graphs.}
\else{}
\paragraph{Temporal Graphs.}
\fi{}
A temporal graph~$\TG$ is a tuple~$(V,\TE,\tau)$ consisting of the set~$V$ of vertices, 
the set~$\TE$ of temporal edges, 
and a discrete time-horizon~$\tau$.
A temporal edge~$e$ is an element in~$\binom{V}{2}\times \set{\tau}$.
Equivalently, 
a temporal graph~$\TG$ can be defined as a vector of static graphs~$(G_1,\ldots,G_\tau)$, 
where each graph is defined over the same vertex set~$V$.
We also denote by~$V(\TG)$, $\TE(\TG)$, and~$\tau(\TG)$ the set of vertices, the set of temporal edges, and the discrete time-horizon of~$\TG$, respectively.
The \emph{underlying graph}~$\UG=\UG(\TG)$ of a temporal graph~$\TG$ is the static graph with vertex set~$V(\TG)$ and 
edge set~$\{e \mid \exists t\in\set{\tau(\TG)}: (e,t)\in \TE(\TG)\}$.

\ifarxiv{}
  \subparagraph{Parameterized Complexity Theory.}
  \else{}
  \paragraph{Parameterized Complexity Theory.}
  \fi{}

  Let~$\Sigma$ be a finite alphabet.
  A parameterized problem~$L$ is a subset~$L\subseteq \{(x,k)\in\Sigma^*\times \N_0\}$.
  An instance~$(x,k)\in\Sigma^*\times \N_0$ is a \yes-instance of~$L$ if and only if~$(x,k)\in L$ (otherwise, it is a \no-instance).
  Two instances~$(x,k)$ and~$(x',k')$ of parameterized problems~$L,L'$ are \emph{equivalent} if~$(x,k)\in L \iff (x',k')\in L'$.
  A parameterized problem~$L$ is fixed-parameter tractable (FPT) if for every input~$(x,k)$ one can decide whether~$(x,k)\in L$ in~$f(k)\cdot |x|^{O(1)}$~time, where~$f$ is some computable function only depending on~$k$.  
  A parameterized problem~$L$ is in~$\XP$ if for every instance~$(x,k)$ one can decide whether~$(x,k)\in L$ in time~$|x|^{f(k)}$ for some computable function~$f$ only depending on~$k$.
  A \W{1}-hard parameterized problem is fixed-parameter intractable unless~\FPT=\W{1}.

  Given a parameterized problem~$L$, a \emph{kernelization} is an algorithm that maps any instance~$(x,k)$ of~$L$ in time polynomial in~$|x|+k$ to an instance~$(x',k')$ of~$L$ (the problem \emph{kernel}) such that
  \begin{inparaenum}[(i)]
  \item $(x,k)\in L \iff (x',k')\in L$, and
  \item $|x'|+k'\leq f(k)$ for some computable function~$f$ (the \emph{size} of the problem kernel) only depending on~$k$.
  \end{inparaenum}

  We refer to Downey~and~Fellows~\cite{downey2013fundamentals} and Cygan~\etal~\cite{cygan2015parameterized} for more material on parameterized complexity.

\section{Basic Observations}
\label{sec:firstobs}

In this section,
we state some preliminary simple-but-useful observations on \msvc{} and its relation to~\vc{}.

\begin{observation}%
 \label[observation]{obs:tau2oneedgelayer}
 Every instance~$(\TG,k,\ell)$ of~\msvc{} with~$k\geq \sum_{i=1}^{\tau(\TG)} |E(G_i)|$ is a \yes-instance.
\end{observation}

\begin{proof}
 It is easy to see that a graph with~$m$ edges always admits a vertex cover of size~$m$.
 Hence, there is a vertex cover~$S\subseteq V$ of size~$k$ of~$\UG(\TG)$, and hence, $S$ is a vertex cover for each layer.
 The vector~$(S_1,\ldots,S_\tau)$ with~$S_i=S$ for all~$i\in\set{\tau}$ is a solution for every~$\ell\geq 0$.
\lqed
\end{proof}

\noindent
Next,
we state that if we are facing a \yes-instance,
then we can assume that there exists a solution where each layer's vertex cover is either of size~$k$ or~$k-1$.

\begin{observation}%
 \label[observation]{obs:largesolutions}
 Let $(\TG,k,\ell)$ be an instance of~\msvc{}. 
 If $(\TG,k,\ell)$ is a \yes-instance,
 then there is a solution $\calS = (S_1,\ldots,S_\tau)$ such that $|S_1| = k$ and $k-1 \leq |S_i| \leq k$ for all $i \in \{1,\ldots,\tau\}$.
\end{observation}
\begin{proof}
	We first show that there is a solution $\calS = (S_1,\ldots,S_\tau)$ for~$I\ceq (\TG,k,\ell)$ such that $|S_1| = k$.
	Towards a contradiction assume that such a solution does not exist.
	Let~$\calS=(S_1,\ldots,S_\tau)$ be a solution such that $|S_1|$~is maximal over all solutions for~$I$.
	Let $i \in \{1,\ldots,\tau\}$ be the maximum index such that~$S_{j} \subseteq S_{j-1}$, for all $j \in \{2,\ldots,i\}$.
	If $i=\tau$, then we have that $|S_j| \leq |S_1| < k$ for all $j \in \{1,\ldots,\tau\}$.
	Hence, 
	we can find a subset~$X\subseteq V \setminus S_1$ such that $(S_1 \cup X,\ldots,S_\tau\cup X)$ is a solution.
	This contradicts $|S_1|$ being maximal.
	Now let $i < \tau$. 
	Hence, there is a vertex $v \in S_{i+1} \setminus S_i$.
	Now we can adjust the solution by adding $v$ to $S_j$ for all $j \in \{1,\dots,i\}$.
	This contradicts $|S_1|$ being maximal.
	Hence, there is a solution $\calS = (S_1,\ldots,S_\tau)$ such that~$|S_1| = k$.

	Let $\solset$ be the set of solutions such that the first vertex cover is of size~$k$.
	Assume towards a contradiction that all solutions in $\solset$ contain a vertex cover smaller than~$k-1$.
	Let $\solset_i \subseteq \solset$ be the set of solutions such that for each~$(S_1,\ldots,S_\tau) \in \solset_i$ we have that~$|S_i| < k-1$ and $|S_j| \geq k-1$ for all $j \in \{1,\ldots, i-1\}$.
	Let $i \in \{1,\ldots,\tau\}$ be maximal such that $\solset_i \not = \emptyset$.
	Furthermore, 
	let $\calS = (S_1,\ldots,S_\tau) \in \solset_i$ such that~$|S_i|$ is maximal over all solutions in~$\solset_i$. 
	Hence, 
	there is a vertex $v \in S_{i-1} \setminus S_i$.
	We distinguish two cases.
	\begin{compactdesc}
	\item[(a):]
	Assume that there is a $p \in \{i+1,\dots,\tau\}$ such that 
	there is a~$w \in S_{p} \setminus S_{p-1}$ and~$S_{j} \subseteq S_{j-1}$ for all~$j \in \{i+1,\ldots, p-1\}$.
	The idea now is to keep~$v$ and add~$w$ in the \ith{$i$} layer and then remove~$v$ in the \ith{$p$} layer.
	We can achieve this by simply setting $S_q \ceq S_q \cup \{v,w\}$ for all $q \in \{i,\ldots,p-1\}$. 
	Note that this is a solution which either contradicts that $|S_i|$ is maximal or that $i$ is maximal.
	
	\item[(b):]
	Now assume that $S_{j} \subseteq S_{j-1}$ for all $j \in \{i+1,\ldots, \tau\}$.
	In this case we take an arbitrary vertex $w \in V \setminus S_i$ and set $S_q \ceq \{v,w\}$ for all $q \in \{i, \ldots, \tau\}$. 
	This contradicts $i$ being maximal.
    \lqed
	\end{compactdesc}
\end{proof}

\noindent
With the next two observations,
we show that the special case of \msvc{} where~$\ell=0$
is equivalent to~\vc{} under polynomial-time many-one reductions.

\begin{observation}%
 \label[observation]{prop:oneedgelayers-ell-0}
 There is a polynomial-time algorithm that maps 
 any instance~$(G=(V,E),k)$ of~\vc{} to 
 an equivalent instance~$(\TG,k,\ell)$ of~\MSVC{} 
 where~$\ell=0$ and every layer~$G_i$ contains only one edge.
\end{observation}

\begin{proof}
 Let the edges~$E=\{e_1,\ldots,e_m\}$ of~$G$ be ordered in an arbitrary way.
 Set~$\tau=m$ and $\ell=0$.
 Set~$G_{i}=(V,\{e_i\})$ for each~$i\in\set{\tau}$.
 We claim that $(G=(V,E),k)$ is a \yes-instance of \vc{} if and only if $(\TG,k,\ell)$ is a \yes-instance of \MSVC{}.
 
 \RD{}
 Let~$S$ be a vertex cover of~$G$ of size at most~$k$.
 Set~$S_i\ceq S$ for all~$i\in\set{\tau}$.
 Clearly,~$S_i$ is a vertex cover of~$G_i$ for all~$i\in\set{\tau}$ of size at most~$k$.
 Moreover, 
 by construction,
 $\sydic{S_i}{S_{i+1}}\leq 1$ for all~$i\in\set{\tau-1}$.
 Hence, 
 $(S_1,\ldots,S_{\tau})$~forms a solution to~$(\TG,k,\ell)$.
 
 \LD{}
 Let  $\calS=(S_1,\ldots,S_{\tau})$ be a solution to~$(\TG,k,\ell)$.
 Observe that~$|\bigcup_i S_i|\leq k$.
 It follows that there are at most~$k$ vertices covering all edges of the layers~$G_{i}$, 
 that is,
 $E=\bigcup_{i=1}^\tau E(G_i)$,
 and hence they cover all edges of~$G$.
\lqed
\end{proof}

\begin{observation}%
 \label[observation]{obs:elleqzero}
 There is a polynomial-time algorithm that maps any instance~$(\TG,k,\ell)$ of \MSVC{} with~$\ell=0$ 
 to an equivalent instance~$(G,k)$ of~\vc{}.
\end{observation}

\begin{proof}
  Now let~$(\TG=(V,\TE,\tau),k,0)$ be an arbitrary instance of~\MSVC{}.
  Construct the instance~$(\UG,k)$ of \vc{}.
  We claim that $(\TG,k,0)$ is a \yes-instance if and only if $(\UG,k)$ is a \yes-instance.
  
  \LD{}
  Let~$S\subseteq V$ be a vertex cover of size at most~$k$.
  Since~$S$ is a vertex cover for~$\UG$,~$S$ covers each layer of~$\TG$.
  Hence,~$S_i\ceq S$ for all~$i\in\set{\tau}$ forms a solution to  $(\TG,k,0)$.
  
  \RD{}
  Let~$(S_1,\ldots,S_\tau)$ be a solution to $(\TG,k,0)$.
  Clearly, since~$\ell=0$, we have that~$S_i=S_j$ for all~$i,j\in\set{\tau}$.
  It is not difficult to see that~$S\ceq  S_1$ is a vertex cover for~$\UG$, and hence the claim follows.
\lqed
\end{proof}

\noindent
Finally,
the special case of~\MSVC{} with~$\ell\geq 2k$ 
(that is, where vertex covers of any two consecutive layers can be even disjoint)
is Turing-reducible to~\vc{}.

\begin{observation}%
 \label[observation]{obs:turedu}
 Any instance~$(\TG,k,\ell)$ of~\MSVC{} with~$\ell\geq 2k$ and~$\TG=(G_1,\ldots,G_\tau)$ can be decided by deciding each instance of the set~$\{(G_i,k)\mid 1\leq i\leq \tau\}$ of \vc{}-instances. 
\end{observation}

\begin{proof}
 For each of the layers~$G_i$, 
 $i\in\set{\tau}$,
 we can construct an instance of \VC{} of the form~$(G_i,k)$.
 We can solve each instance independently, 
 since the symmetric difference of any two size-at-most-$k$ solutions is at most~$2k\leq \ell$.
\lqed
\end{proof}

\section{Hardness for Restricted Input Instances}
\label{sec:hardness}

\msvc{} is \NP-hard as it generalizes \vc{} ($\tau=1$). 
In this section we prove that \MSVC{} remains \NP-hard on 
inputs with only two layers,
one consisting of a path and the other consisting of a tree, 
and on inputs where every layer consists only of one edge.

\begin{theorem}
 \label{thm:npahrdcases}
 \msvc{} is \NP-hard even if 
 \begin{compactenum}[(i)]
  \item~$\tau=2$, $\ell=0$, and the first layer is a path and the second layer is a tree, or
  \item~every layer contains only one edge and~$\ell \leq 1$.
 \end{compactenum}
\end{theorem}

\begin{remark}
\cref{thm:npahrdcases}(i) is tight regarding~$\tau$ since \vc{} 
(i.e.,\ \MSVC{} with $\tau=1$) on trees is solvable in linear time.
\cref{thm:npahrdcases}(ii) is tight regarding~$\ell$, because in the case of $\ell>1$  \cref{obs:turedu} is applicable.
\end{remark}
\vc{} remains \NP-complete 
on cubic Hamiltonian graphs when a
Hamiltonian cycle is additionally given in the input
\cite{FleischnerSS10}:%
\footnote{A graph is cubic if each vertex is of degree exactly three; 
A graph is Hamiltonian if it contains a subgraph being a Hamiltonian cycle, 
that is, 
a cycle that visits each vertex in the graph exactly once.}

\problemdef{Hamiltonian Cubic Vertex Cover (HCVC)}
{An undirected, cubic, Hamiltonian graph~$G=(V,E)$,
an integer~$k\in\N$, 
and a Hamiltonian Cycle~$C=(V,E')$ of~$G$.}
{Is there a set~$S\subseteq V$ such that~$S$ is a size-at-most-$k$ vertex cover for~$G$?}
To prove \cref{thm:npahrdcases}(i), 
we give a polynomial-time many-one reduction from HCVC to \MSVC{} with two layers, one being a path, the other being a tree.

\begin{proposition}%
 \label{prop:tau2treelayers}
 There is a polynomial-time algorithm that 
 maps any instance~$(G=(V,E),k,C)$ of \prob{HCVC} to 
 an equivalent instance~$(\TG,k',\ell')$ of~\MSVC{} 
 with~$\tau=2$ and the first layer~$G_1$ being a path and second layer~$G_2$ being a tree.
\end{proposition}

\begin{proof}
  Let~$e\in E(C)$ be some edge of~$C$, 
  and let~$P=C-e$ be the Hamiltonian path obtained from~$C$ when removing~$e$.
  Let~$E_1\ceq E(P)$, 
  and~$E_2\ceq E\setminus E(P)$.
  Set initially~$G_1=(V,E_1)$ and~$G_2=(V,E_2)$.
  Note that~$G_1$ is a path.
  Moreover,
  observe that~$G_2$ is the disjoint union of~$|V|/2-2$ paths of length one 
  and one path of length three:
  the graph~$G-E(C)$ is a disjoint union of~$|V|/2$ paths of length one,
  since each vertex is of degree three in~$G$,
  is adjacent to two vertices in~$C$,
  and thus has degree one in~$G-E(C)$;
  Since~$G-E(C)=G_2-e$,
  edge~$e$ connects two paths of length one to one path of length three in~$G_2$.
  Add two special vertices~$z,z'$ to~$V$.
  In~$G_1$,
  connect~$z$ with~$z'$ and with one endpoint of~$P$.
  In~$G_2$,
  connect~$z$ with~$z'$ and with exactly one vertex of each connected component.
  Set~$k'=k+1$ and~$\ell'=0$.
  We claim that $(G=(V,E),k,C)$ is a \yes-instance if and only if $(\TG,k',\ell')$ is a \yes-instance.

  \RD{}
  Let~$S'$ be a vertex cover of~$G$ of size at most~$k$.
  We claim that~$S'\ceq S\cup\{z\}$ is a vertex cover for both~$G_1$ and~$G_2$.
  Observe that~$G_1[E_1]$ and~$G_2[E_2]$ are subgraphs of~$G$, and hence all edges are covered by~$S'$.
  Moreover, 
  all edges in~$G_i-E_i$, $i\in\{1,2\}$, are incident with~$z$ and hence covered by~$S'$.
  
  \LD{}
  Let~$(S_1,S_2)$ be a minimal solution to $(\TG,k',\ell')$ with~$S'\ceq S_1=S_2$ and~$|S'|\leq k'$.
  We can assume that~$z\in S'$ since the edge~$\{z,z'\}$ is present in both~$G_1$ and~$G_2$, and exchanging~$z$ in~$z'$ does not cover less edges.
  Moreover, we can assume that not both~$z$ and~$z'$ are in~$S'$ due to the minimality of~$S'$.
  Let~$S\ceq S'\setminus\{z\}$.
  Observe that~$S$ covers all edges in~$E_1\cup E_2$ and,
  hence, 
  $S$~forms a vertex cover of~$G$ of size at most~$k=k'-1$.
\lqed
\end{proof}

Note that \cref{thm:npahrdcases}(ii) for $\ell=0$ is already shown by \cref{prop:oneedgelayers-ell-0}.
In order to prove~\cref{thm:npahrdcases}(ii) for $\ell=1$, 
we adjust the polynomial-time many-one reduction behind \cref{prop:oneedgelayers-ell-0}.

\begin{proposition}%
 \label{prop:oneedgelayers}
 There is a polynomial-time algorithm that maps 
 any instance~$(G=(V,E),k)$ of~\vc{} to 
 an equivalent instance~$(\TG,k',\ell')$ of~\MSVC{} 
 where~$\ell'=1$ and every layer~$G_i$ contains only one edge.
\end{proposition}

\begin{proof}
 Let the edges~$E=\{e_1,\ldots,e_m\}$ of~$G$ be arbitrarily ordered.
 Set~$\tau=2m$.
 Set~$V'=V\cup W$, where~$W=\{w_1,\ldots,w_{2\tau}\}$.
 Set~$G_{2i-1}=(V',\{e_i\})$ and~$G_{2i}=(V',\{w_i,w_{i+\tau}\})$ for each~$i\in\set{\tau}$.
 Set~$k'=k+1$ and~$\ell=1$.
 We claim that $(G=(V,E),k)$ is a \yes-instance of \vc{} if and only if $(\TG,k',\ell')$  is a \yes-instance of \MSVC{}.
 
 \RD{}
 Let~$S$ be a vertex cover of~$G$ of size at most~$k$.
 Set~$S_{2i-1}\ceq S$, and~$S_{2i}\ceq S\cup\{w_i\}$ for all~$i\in\set{\tau}$.
 Clearly,~$S_i$ is a vertex cover of~$G_i$ for all~$i\in\set{2\tau}$ of size at most~$k'=k+1$.
 Moreover, 
 by construction,
 $\sydic{S_i}{S_{i+1}}\leq 1$ for all~$i\in\set{2\tau-1}$.
 Hence, 
 $(S_1,\ldots,S_{2\tau})$~forms a solution to~$(\TG,k',\ell')$.
 
 \LD{}
 Let  $\calS=(S_1,\ldots,S_{2\tau})$ be a solution to~$(\TG,k',\ell')$.
 Observe that~$|\bigcup_i S_i|\leq k+\tau$.
 We know that~$|W\cap \bigcup_i S_i|\geq \tau$.
 It follows that there are at most~$k$ vertices covering all edges of the layers~$G_{2i-1}$, 
 that is, 
 $E=\bigcup_{i=1}^\tau E(G_{2i-1})$,
 and,
 hence,
 covering all edges of~$G$.
\lqed
\end{proof}

\cref{thm:npahrdcases} now follows from \cref{prop:tau2treelayers,prop:oneedgelayers}.

\section{Parameter Vertex Cover Size}
\label{sec:paramvc}

In this section, we study the parameter size~$k$ of the vertex cover of each layer for \MSVC{}.
\vc{} and \prob{Vertex Cover Reconfiguration}~\cite{mouawad2017parameterized} when parameterized by the vertex cover size are fixed-parameter tractable. 
We prove that this is no longer true for \MSVC{} (unless~$\FPT{}=\W{1}$).
\begin{theorem}
  \label{thm:xpwhardness}
 \msvc{} parameterized by~$k$ is in \XP{} and \W{1}-hard.
\end{theorem}
We first show the XP-algorithm (\cref{ssec:xpalgo}),
and then prove the \W{1}-hardness (\cref{ssec:whardness}) and discuss its implications.

\subsection{XP-Algorithm}
\label{ssec:xpalgo}

In this section, we prove the following.
\begin{proposition}
 \label{prop:xpalgo}
 Every instance~$(\TG,k,\ell)$ of \msvc{} can be decided in~$O(\tau(\TG)\cdot |V(\TG)|^{2k+1})$ time.
\end{proposition}

In a nutshell, 
to prove~\cref{prop:xpalgo} we first consider for each layer all vertex subsets of size at most~$k$
that form a vertex cover.
Second, we find a sequence of vertex covers for all layers such that the sizes of the symmetric differences for every two consecutive solutions is at most~$\ell$.
We show that the second step can be solved via computing a source-sink path in an auxiliary directed graph that we call \emph{configuration graph} (see~\cref{fig:xp} for an illustrative example). 
\begin{definition}
	Given an instance~$I=(\TG,k,\ell)$ of \MSVC{}, 
	the \emph{configuration graph} of~$I$ is 
	the directed graph~$ D=(V,A,\gamma) $ 
	with~$V=V_1\uplus\cdots\uplus V_\tau\uplus\{s,t\}$,
	being equipped with a function~$\gamma:V\to \{V'\subseteq V(\TG)\mid |V'|\leq k\}$ such that
 \begin{compactenum}[(i)]
  \item for every~$i\in\set{\tau(\TG)}$, it holds true that $S$ is a vertex cover of~$G_i$ of size exactly~$k-1$ or~$k$ if and only if there is a vertex~$v\in V_i$ with $\gamma(v)=S$,
  \item there is an arc from~$v\in V$ to~$w\in V$ if and only if~$v\in V_i$, $w\in V_{i+1}$, and~$\sydic{\gamma(v)}{\gamma(w)}\leq \ell$, and
  \item there is an arc~$(s,v)$ for all~$v\in V_1$ and an arc~$(v,t)$ for all~$v\in V_\tau$.
 \end{compactenum}
\end{definition}

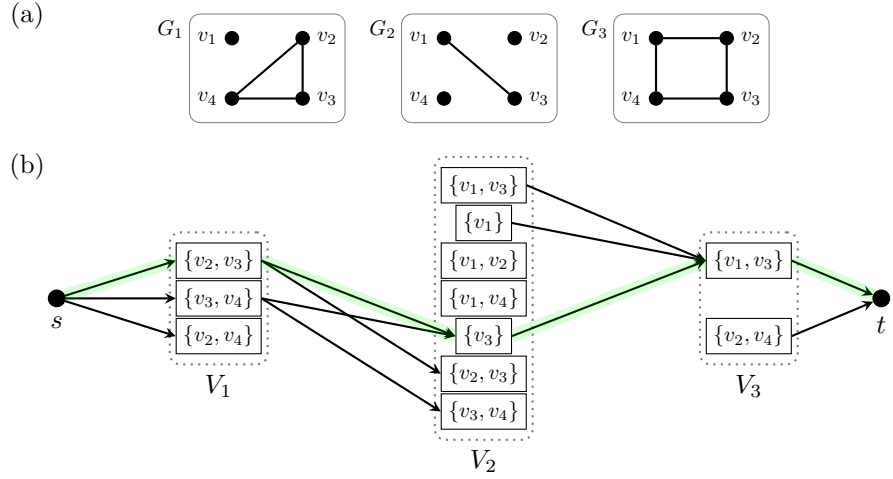
\begin{figure}[t]
    \centering
    \begin{tikzpicture}
      \def\xr{0.775}
      \def\yr{1}
      \tikzstyle{xnode}=[circle,fill,scale=2/3,draw];
      \tikzstyle{cnode}=[scale=0.8,draw];
      \tikzstyle{xedge}=[thick,-];
      \tikzstyle{xarc}=[thick,-stealth];
      \tikzstyle{harc}=[line width=5pt,color=green,opacity=0.15];
      \tikzstyle{vertexbox}=[dotted,gray,thick,rounded corners,draw];
      \newcommand{\tikzlayer}[1]{
          \draw[rounded corners, gray, very thin] (0.1*\xr,0.1*\yr) rectangle (3.4*\xr,1.9*\yr);
          \node at (-0.3*\xr,1.25*\yr)[label=90:{#1}]{};
          \node (v1) at (1*\xr,1.5*\yr)[xnode,label=180:{$v_1$}]{};
          \node (v2) at (2.5*\xr,1.5*\yr)[xnode,label=0:{$v_2$}]{};
          \node (v3) at (2.5*\xr,0.5*\yr)[xnode,label=00:{$v_3$}]{};
          \node (v4) at (1*\xr,0.5*\yr)[xnode,label=180:{$v_4$}]{};
      }

      \begin{scope}[xshift=\xr*1.175cm,scale=0.8,transform shape]
          
          \begin{scope}[]
              \tikzlayer{$G_1$}
              \draw[xedge] (v2) -- (v3) -- (v4) -- (v2);		
          \end{scope}

          \begin{scope}[xshift=\xr*4.5cm]
              \tikzlayer{$G_2$}
              \draw[xedge] (v3) -- (v1);
          \end{scope}

          \begin{scope}[xshift=\xr*9cm]
              \tikzlayer{$G_3$}
              \draw[xedge] (v2) -- (v3) -- (v4) -- (v1) -- (v2);
          \end{scope}
      \end{scope}
      \node at (-1.5*\xr,1.5*\yr)[]{(a)};
      
      \def\ys{2.25};
      \def\ysn{0.5};
      \node (s) at (-1*\xr,-\ys*\yr)[xnode,label=-90:{$s$}]{};
      \node (t) at (13*\xr,-\ys*\yr)[xnode,label=-90:{$t$}]{};

      \begin{scope}[yshift=-\yr*\ys cm]
          \node (a1) at (1.75*\xr,1*\ysn*\yr)[cnode]{$\{v_2,v_3\}$};
          \node (a2) at (1.75*\xr,0*\yr)[cnode]{$\{v_3,v_4\}$};
          \node (a3) at (1.75*\xr,-1*\ysn*\yr)[cnode]{$\{v_2,v_4\}$};
          \node at (a2)[minimum width=\xr*1.66cm, minimum height=\yr*1.75cm,vertexbox,label=-90:{$V_1$}]{};
      \end{scope}

      \begin{scope}[xshift=\xr*4.5cm,yshift=-\yr*\ys cm]
          \node (b1) at (1.75*\xr,3*\ysn*\yr)[cnode]{$\{v_1,v_3\}$};
          \node (b2) at (1.75*\xr,2*\ysn*\yr)[cnode]{$\{v_1\}$};
          \node (b3) at (1.75*\xr,1*\ysn*\yr)[cnode]{$\{v_1,v_2\}$};
          \node (b4) at (1.75*\xr,0)[cnode]{$\{v_1,v_4\}$};
          \node (b5) at (1.75*\xr,-1*\ysn*\yr)[cnode]{$\{v_3\}$};
          \node (b6) at (1.75*\xr,-2*\ysn*\yr)[cnode]{$\{v_2,v_3\}$};
          \node (b7) at (1.75*\xr,-3*\ysn*\yr)[cnode]{$\{v_3,v_4\}$};
          \node at (b4)[minimum width=\xr*1.66cm, minimum height=\yr*3.75cm,vertexbox,label=-90:{$V_2$}]{};
      \end{scope}
      
      \begin{scope}[xshift=\xr*9cm,yshift=-\yr*\ys cm]
          \node (c1) at (1.75*\xr,1*\ysn*\yr)[cnode]{$\{v_1,v_3\}$};
    \node (c0) at (1.75*\xr,0*\yr)[draw=none]{};
          \node (c2) at (1.75*\xr,-1*\ysn*\yr)[cnode]{$\{v_2,v_4\}$};
          \node at (c0)[minimum width=\xr*1.66cm, minimum height=\yr*1.75cm,vertexbox,label=-90:{$V_3$}]{};
      \end{scope}

      \foreach \x in {1,2,3}{\draw[xarc] (s) to (a\x.west);}
      \foreach \x in {5,6}{\draw[xarc] (a1.east) to (b\x.west);}
      \foreach \x in {5,7}{\draw[xarc] (a2.east) to (b\x.west);}
      \foreach \x in {1,2,5}{\draw[xarc] (b\x.east) to (c1.west);}
      \foreach \x in {1,2}{\draw[xarc] (c\x.east) to (t);}
      \draw[harc] (s) to (a1.west);
      \draw[harc] (a1.east) to (b5.west);
      \draw[harc] (b5.east) to (c1.west);
      \draw[harc] (c1.east) to (t);

      \node at (-1.5*\xr,-0.5*\yr)[]{(b)};
    \end{tikzpicture}
    \caption{Illustrative example of a configuration graph. 
    (a) Temporal graph instance~$I=(\TG,k,\ell)$ from~\cref{fig:introex} with~$\TG=(G_1,G_2,G_3)$, 
    $k=2$, 
    and~$\ell=1$. 
    (b) Configuration graph of~$I$ from (a); a directed~$s$-$t$ path is highlighted corresponding to the solution depicted in~\cref{fig:introex}.}
    \label{fig:xp}
\end{figure}

\noindent
Note that Mouawad~\etal~\cite{mouawad2017parameterized} used a similar 
configuration graph 
to show fixed-parameter tractability of \prob{Vertex Cover Reconfiguration} parameterized by the vertex cover size $k$.
In the multistage setting the configuration graph is 
too large for fixed-parameter tractability regarding $k$.
However, we show an XP-algorithm regarding~$k$ 
to construct the configuration graph.
\begin{lemma}%
 \label{lem:confgraph}
 The configuration graph of an instance~$(\TG,k,\ell)$ of \MSVC{}, 
 where temporal graph~$\TG$ has~$n$ vertices and time horizon~$\tau$,
\begin{compactenum}[(i)]
    \item can be constructed in~$O(\tau \cdot n^{2k+1})$ time, and 
    \item contains at most $\tau\cdot 2n^k+2$ vertices and~$(\tau-1) n^{2k}+4n^k$ arcs.
 \end{compactenum}
\end{lemma}
\begin{proof}
	Compute the set~$\calS=\{V'\subseteq V(\TG) \mid k-1 \leq |V'|\leq k\}$ in~$O(n^k)$ time.
 For each layer~$G_i$ and each set $S\in\calS$, check in $O(|E(G_i)|)$~time whether $S$ is a vertex cover for~$G_i$.
 Let~$\calS_i\subseteq \calS$ denote the set of vertex covers of size~$k-1$ or $k$ of layer~$G_i$.
 For each~$S\in\calS_i$, add a vertex~$v$ to~$V_i$ and set~$\gamma(v)=S$.
 Lastly,
 add the vertices~$s$ and~$t$.
 Hence, 
 we can construct the vertex set~$V$ of the configuration graph~$D$ of size~$\tau\cdot 2n^k+2$ in~$O(n^{k+2} \cdot \tau)$~time.
 For every~$i\in\set{\tau-1}$, 
 and every~$v\in V_i$ and~$w\in V_{i+1}$,
 check whether~$\sydic{\gamma(v)}{\gamma(w)}\leq \ell$ in~$O(k)$ time.
 If this is the case, 
 then add the arc~$(v,w)$.
 The latter steps can be done in~$O(n^{2k+1}\cdot (\tau-1)) $ time.
 Finally, 
 add the arc~$(s,v)$ for each~$v\in V_1$ and the arc~$(v,t)$ for each~$v\in V_\tau$ in~$O(n^k)$~time.
 The finishes the construction of~$D=(V=V_1\uplus\cdots\uplus V_\tau\uplus\{s,t\} ,A,\gamma)$.
 \lqed
\end{proof}

The crucial observation is that we can decide any instance by checking for an $s$-$t$ path in its configuration graph.

 \begin{lemma}%
 \label{lem:stpathinconfgraph}
 \MSVC{}-instance~$I=(\TG,k,\ell)$ is a \yes-instance if and only if 
 there is an~$s$-$t$ path in the configuration graph~$D$ of~$I$.
\end{lemma}

\begin{proof}
	Let $D=(V=V_1\uplus\cdots\uplus V_\tau\uplus\{s,t\} ,A,\gamma)$.

 \RD{}
 Let~$(S_1,\ldots,S_\tau)$ be a solution to~$(\TG,k,\ell)$.
 By \cref{obs:largesolutions}, we can assume without loss of generality that $k-1 \leq |S_i| \leq k$, for all $i \in \{1,\dots,\tau\}$.
 Hence for each~$S_i$, there is a~$v_i\in V_i$ such that $\gamma(v_i)=S_i$, for all $i \in \set{\tau}$.
 Note that the arc~$(v_i,v_{i+1})$ is contained in~$A$ for each~$i\in\set{\tau-1}$ since~$\sydic{\gamma(v_i)}{\gamma(v_{i+1})}=\sydic{S_i}{S_{i+1}}\leq \ell$.
 Hence,~$P=(\{v_1,\ldots,v_\tau\}\cup\{s,t\},\{ (s,v_1), (v_\tau,t)\} \cup\bigcup_{i=1}^{\tau-1}\{ (v_i,v_{i+1})\})$ is an~$s$-$t$ path in~$D$.
 
 \LD{}
 Let $P=(\{v_1,\ldots,v_\tau\}\cup\{s,t\},\{(s,v_1),(v_\tau,t)\}\cup\bigcup_{i=1}^{\tau-1}\{(v_i,v_{i+1})\})$ be an~$s$-$t$ path in~$D$.
 We claim that~$(\gamma(v_i))_{i\in\set{\tau}}$ forms a solution to~$(\TG,k,\ell)$.
 First, note that for all~$i\in\set{\tau}$,~$\gamma(v_i)$ is a vertex cover for~$G_i$ of size at most~$k$.
 Moreover, for all~$i\in\set{\tau-1}$, $\sydic{\gamma(v_i)}{\gamma(v_{i+1})}\leq \ell$ since the arc~$(v_i,v_{i+1})$ is present in~$D$. 
 This finishes the proof.
\lqed
\end{proof}

\noindent
We are ready to prove \cref{prop:xpalgo}. 

\ifarxiv{}
\begin{proof}[Proof of~\cref{prop:xpalgo}]
\else{}
\begin{proof}[of~\cref{prop:xpalgo}]
\fi{}
	First, compute the configuration graph~$D$ of the instance $(\TG=(V,\TE,\tau),k,\ell)$ of \msvc{} in~$O(\tau\cdot|V|^{2k+1})$ time (\cref{lem:confgraph}(i)).
	Then, find an $s$-$t$ path in~$D$ with a breadth-first search in~$O(\tau\cdot|V|^{2k})$ time (\cref{lem:confgraph}(ii)).
	If an~$s$-$t$ path is found, then return~\yes, otherwise return \no{} (\cref{lem:stpathinconfgraph}).
\lqed
\end{proof}

\begin{remark}
  The reason why the algorithm behind \cref{prop:xpalgo} 
  is only an \XP-algorithm and not 
  an \FPT-algorithm regarding~$k$ 
  is because we do not have a better upper~bound on the number of vertices in the 
  configuration graph for $(\TG,k,\ell)$ than $O(\tau(\TG) \cdot |V(\TG)|^k)$.
  This is due to the fact that we check for each subset of $V(\TG)$ of size $k$ or $k-1$ whether it is a vertex cover in some layer.

  This changes if we consider \prob{Minimal \msvc{}} 
  where we additionally demand the $i$-th set in the solution to be a \emph{minimal} vertex cover for the layer $G_i$.
  Here, we can enumerate for each layer~$G_i$ all minimal vertex covers of size at most $k$ 
  (and hence all candidates for the $i$-th set of the solution)
  with the folklore search-tree algorithm for vertex cover. %
  This leads to $O(2^k\tau(\TG))$ many vertices in the configuration graph (for \prob{Minimal \msvc{}}) and thus to fixed-parameter tractability of \prob{Minimal \msvc{}} parameterized by the vertex cover size $k$.
\end{remark}

However, it is unlikely (unless \FPT$=$\W{1}) 
that one can substantially improve the algorithm behind \cref{prop:xpalgo},
as we show next.

\subsection{Fixed-parameter Intractability}
\label{ssec:whardness}

In this section we show that \MSVC{} is \W{1}-hard when parameterized by~$k$.
This hardness result is established by the following parameterized reduction from the W[1]-complete~\cite{DowneyF99} \prob{Clique} problem, 
where, given an undirected graph~$G$ and a positive integer~$k$, 
the question is whether~$G$ contains a clique of size~$k$ (that is, $k$ vertices that are pairwise adjacent).
\begin{proposition}
  \label{prop:whardness}
 There is an algorithm that maps any instance~$(G,k)$ of~\prob{Clique} 
 in polynomial time 
 to an equivalent instance~$(\TG,k',\ell)$ of~\MSVC{} 
 with~$k'=2\binom{k}{2}+k+1$,~$\ell=2$, 
 and each layer of~$\TG$ being a forest with $O(k^4)$ edges.
\end{proposition}

\noindent
In the remainder of this section,
we prove \cref{prop:whardness}.
We next give the construction of the \MSVC{} instance,
then prove the forward (\cref{sssec:fd}) and backward (\cref{sssec:bd}) direction of the equivalence,
and finally (in~\cref{sssec:proofwhardness}) put the pieces together and derive two corollaries.

We construct an instance of \msvc{} from an instance 
of \prob{Clique} as follows (see \cref{fig:whardness} for an illustrative example).

\begin{construction}
  \label[construction]{constr:whardness}
  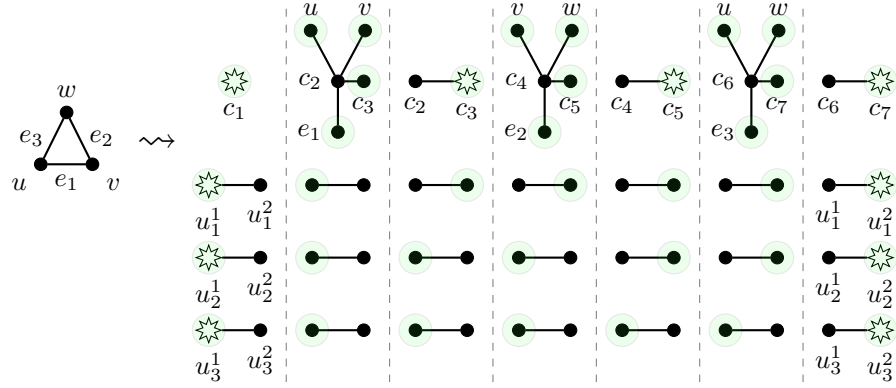
\begin{figure}
   \centering
   \begin{tikzpicture}[node distance=0.5cm and 0.5cm]
    \usetikzlibrary{decorations,shapes,positioning,calc}

    \tikzstyle{snode}=[star,star points=8,star point ratio=2,scale=1/2,fill=white,draw]
    \tikzstyle{xnode}=[circle, fill,scale=1/2,draw]
    \tikzstyle{hnode}=[fill=green!75!white,draw=none,circle,scale=1.3,draw,opacity=0.1];

    \def\xr{0.68}
    \def\yr{0.55}

    \begin{scope}[yshift=-2*\yr cm]
    \node (u) at (0,0)[xnode,label=-135:{$u$}]{};
    \node (v) at (1*\xr,0)[xnode,label=-45:{$v$}]{};
    \node (w) at (0.5*\xr,1.25*\yr)[xnode,label=90:{$w$}]{};
    \draw[thick] (u) to node[below]{$e_1$}(v);
    \draw[thick] (v) to node[right]{$e_2$}(w);
    \draw[thick] (w) to node[left]{$e_3$}(u);
    \node at (2.25*\xr,0.5*\yr)[scale=1.5]{$\leadsto$};
    \end{scope}

    \begin{scope}[xshift=1.75*\xr cm]

        \node (c1) at (1*2*\xr,0)[snode,label=-90:{$c_1$}]{};
        
		\foreach \x in {3,5,7}{
			\node (cs\x) at (\x*2*\xr-2*\xr+2.5*\xr,0)[snode,label=-90:{$c_\x$}]{};
		}
		\foreach \x in {3,5,7}{
			\node (c\x) at (\x*2*\xr-2*2*\xr+2.5*\xr,0)[xnode,label=-90:{$c_\x$}]{};
		}
		\foreach \x in {2,4,6}{	
			\node (c\x) at (\x*2*\xr-2*\xr+2*\xr,0)[xnode,label=-180:{$c_\x$}]{};
		}
		\foreach \x in {2,4,6}{	
			\node (cs\x) at (\x*2*\xr+1.5*\xr,0)[xnode,label=-90:{$c_\x$}]{};
		}
		\foreach \x in {1,2,...,6}{
			\draw[dashed,gray] (1*\xr+\x*2*\xr,1.75*\yr) -- (1*\xr+\x*2*\xr,-7.25*\yr); 
		}
		
		\begin{scope}[yshift=-0.25*\yr cm]
            
            \node (0-u11) at (1.5*\xr,-2.25*\yr)[snode,label=-90:{$u_1^1$}]{};
            \node (0-u1i) at (1.5*\xr,-4.00*\yr)[snode,label=-90:{$u_2^1$}]{};
            \node (0-u1K) at (1.5*\xr,-5.75*\yr)[snode,label=-90:{$u_3^1$}]{};
            
            \node (0-u21) at (2.5*\xr,-2.25*\yr)[xnode,label=-90:{$u_1^2$}]{};
            \node (0-u2i) at (2.5*\xr,-4.0*\yr)[xnode,label=-90:{$u_2^2$}]{};
            \node (0-u2K) at (2.5*\xr,-5.75*\yr)[xnode,label=-90:{$u_3^2$}]{};
            
            \draw[thick] (0-u11) -- (0-u21);
            \draw[thick] (0-u1i) -- (0-u2i);
            \draw[thick] (0-u1K) -- (0-u2K);
            
            \foreach \x in {2,4,6,8,10}{	
                \node (\x-u11) at (\x*\xr+1.5*\xr,-2.25*\yr)[xnode,label=-90:{}]{};
                \node (\x-u1i) at (\x*\xr+1.5*\xr,-4.00*\yr)[xnode,label=-90:{}]{};
                \node (\x-u1K) at (\x*\xr+1.5*\xr,-5.75*\yr)[xnode,label=-90:{}]{};
                
                \node (\x-u21) at (\x*\xr+2.5*\xr,-2.25*\yr)[xnode,label=-90:{}]{};
                \node (\x-u2i) at (\x*\xr+2.5*\xr,-4.00*\yr)[xnode,label=-90:{}]{};
                \node (\x-u2K) at (\x*\xr+2.5*\xr,-5.75*\yr)[xnode,label=-90:{}]{};
                
                \draw[thick] (\x-u11) -- (\x-u21);
                \draw[thick] (\x-u1i) -- (\x-u2i);
                \draw[thick] (\x-u1K) -- (\x-u2K);
            }
            
            \node (12-u11) at (12*\xr+1.5*\xr,-2.25*\yr)[xnode,label=-90:{$u_1^1$}]{};
            \node (12-u1i) at (12*\xr+1.5*\xr,-4.00*\yr)[xnode,label=-90:{$u_2^1$}]{};
            \node (12-u1K) at (12*\xr+1.5*\xr,-5.75*\yr)[xnode,label=-90:{$u_3^1$}]{};

            \node (12-u21) at (12*\xr+2.5*\xr,-2.25*\yr)[snode,label=-90:{$u_1^2$}]{};
            \node (12-u2i) at (12*\xr+2.5*\xr,-4.00*\yr)[snode,label=-90:{$u_2^2$}]{};
            \node (12-u2K) at (12*\xr+2.5*\xr,-5.75*\yr)[snode,label=-90:{$u_3^2$}]{};

            \draw[thick] (12-u11) -- (12-u21);
            \draw[thick] (12-u1i) -- (12-u2i);
            \draw[thick] (12-u1K) -- (12-u2K);
		\end{scope}

		\node[below=of c2] (uv)[xnode,label=-180:{$e_1$}]{};
		\node[above=of c2,xshift=-15*\xr] (u1) [xnode,label=90:{$u$}]{};
		\node[above=of c2,xshift=15*\xr] (v1) [xnode,label=90:{$v$}]{};
		\foreach \x in {u1,v1,uv}{\draw[thick] (c2) -- (\x);	}

		\node[below=of c4] (vw)  [xnode,label=-180:{$e_2$}]{};
		\node[above=of c4,xshift=-15*\xr] (v2)[xnode,label=90:{$v$}]{};
		\node[above=of c4,xshift=15*\xr] (w2) [xnode,label=90:{$w$}]{};
		\foreach \x in {v2,w2,vw}{\draw[thick] (c4) -- (\x);	}

		\node[below=of c6] (wu) [xnode,label=-180:{$e_3$}]{};
		\node[above=of c6,xshift=-15*\xr] (u3) [xnode,label=90:{$u$}]{};
		\node[above=of c6,xshift=15*\xr] (w3) [xnode,label=90:{$w$}]{};
		\foreach \x in {u3,w3,wu}{\draw[thick] (c6) -- (\x);	}
		
        \draw[thick] (c2) -- (c3);
        \draw[thick] (c4) -- (c5);
        \draw[thick] (c6) -- (c7);
        \draw[thick] (cs2) -- (cs3);
        \draw[thick] (cs4) -- (cs5);
        \draw[thick] (cs6) -- (cs7);

        \node at (0-u11)[hnode]{};
        \node at (0-u1i)[hnode]{};
        \node at (0-u1K)[hnode]{};
        \node at (c1)[hnode]{};

        \node at (2-u11)[hnode]{};
        \node at (2-u1i)[hnode]{};
        \node at (2-u1K)[hnode]{};
        \node at (c3)[hnode]{};
        \node at (u1)[hnode]{};
        \node at (v1)[hnode]{};
        \node at (uv)[hnode]{};

        \node at (4-u21)[hnode]{};
        \node at (4-u1i)[hnode]{};
        \node at (4-u1K)[hnode]{};
        \node at (cs3)[hnode]{};

        \node at (6-u21)[hnode]{};
        \node at (6-u1i)[hnode]{};
        \node at (6-u1K)[hnode]{};
        \node at (c5)[hnode]{};
        \node at (v2)[hnode]{};
        \node at (w2)[hnode]{};
        \node at (vw)[hnode]{};

        \node at (8-u21)[hnode]{};
        \node at (8-u2i)[hnode]{};
        \node at (8-u1K)[hnode]{};
        \node at (cs5)[hnode]{};

        \node at (10-u21)[hnode]{};
        \node at (10-u2i)[hnode]{};
        \node at (10-u1K)[hnode]{};
        \node at (c7)[hnode]{};
        \node at (u3)[hnode]{};
        \node at (w3)[hnode]{};
        \node at (wu)[hnode]{};

        \node at (12-u21)[hnode]{};
        \node at (12-u2i)[hnode]{};
        \node at (12-u2K)[hnode]{};
        \node at (cs7)[hnode]{};
        \end{scope}

    \end{tikzpicture}
    \caption{Illustration of~\cref{constr:whardness} on an example graph (left-hand side) and the first seven layers of the obtained graph (right-hand side). 
    Dashed vertical lines separate layers,
    and for each layer all present edges (but only their incident vertices) are depicted.
    Star-shapes illustrate star graphs with~$k'+1$ leaves. 
    Vertices in a solution (layers' vertex covers) are highlighted.
    }
    \label{fig:whardness}
  \end{figure}
 Let~$(G=(V,E),k)$ be an instance of \prob{Clique} with~$m \ceq |E|$ and $E = \{e_1,\ldots, e_m\}$.
 Let
 \begin{align*} 
    &K\ceq\textstyle{\binom{k}{2}},\, &&k'\ceq 2K+k+1, \,&&\text{ and }\quad \kappa \ceq K + k+3.
 \end{align*}  
 We construct a temporal graph~$\TG=(V',\TE,\tau)$ as follows.
 Let~$V'$ be initialized to~$V \cup E$ (note that~$E$ simultaneously describes the edge set of~$G$ and a vertex subset of~$\TG$).
 We add the following vertex sets
 \begin{align*}
  U^t &\ceq \{u_j^t\mid j\in\set{K}\} \text{ for every } t\in\set{\kappa+1}, \text{ and }\\
  C &\ceq \{c_1,\ldots,c_{2m\kappa+1}\} \text{ (we refer to~$C$ as the set of \emph{center} vertices)}. 
 \end{align*}
 \todo[inline]{tf: we should try our best to keep variables and interpretation the same over this whole technical section, i.e.~$i,j,t$ etc.}
 Let $\TE$ be initially empty.
 We extend the set~$V'$ and define $\TE$
 through the $\tau\ceq 2m\kappa+1$ layers we construct in the following.
 \begin{compactenum}[(1)]
  \item In each layer~$G_{i}$ with~$i$ being odd, 
  make~$c_i$ the center of a star with~$k'+1$ leaves.\footnote{A star (graph) is a tree where at most one vertex (so-called center) is of degree larger than one.}
  \item In each layer~$G_{2mj+1}$, $j\in\{0,\ldots,\kappa\}$, make each vertex in~$U^{j+1}$ the center of a star with~$k'+1$ leaves. 
  \item For each~$j\in\{0,\ldots,\kappa-1\}$, in each layer~$G_{2mj+i}$ with~$i\in\set{2m+1}$, make~$u_x^{j+1}$ adjacent to~$u_x^{j+2}$ for each~$x\in\set{K}$.
  \item For each even~$i$, add the edge~$\{c_i,c_{i+1}\}$ to~$G_i$ and to~$G_{i+1}$.\label{const:4}
  \item For each~$j\in\{0,\ldots,\kappa-1\}$, for each~$i\in\set{m}$, in~$G_{2mj+2i}$, make~$c_{j2m+2i}$ adjacent with~$e_i=\{v,w\}$, $v$, and~$w$.
 \end{compactenum}
 This finishes the construction of~$\TG$. 
  \cqed
\end{construction}

\noindent
The construction essentially repeats the same gadget (which we call \emph{phase}) $\kappa$~times, 
where the layer~$2m\cdot i+1$ is simultaneously the last layer of phase~$i$ 
and the first layer of phase~$i+1$.
In the beginning of phase~$i$, 
a solution has to contain the vertices of~$U^i$.
The idea now is that during phase $i$ one has to exchange the vertices of~$U^i$
with the vertices of~$U^{i+1}$.

It is not difficult to see that the instance in~\cref{constr:whardness} can be computed in polynomial time.
Hence, 
it remains to prove the equivalence stated in~\cref{prop:whardness}.
Recall that we prove the forward and the backward direction in~\cref{sssec:fd,sssec:bd}, 
respectively,
and finally prove~\cref{prop:whardness} in~\cref{sssec:proofwhardness}.

\subsubsection{Forward Direction}
\label{sssec:fd}

The forward direction of \cref{prop:whardness} is---in a nutshell---as follows:
If~$V'\cup E'$ with~$V'\subseteq V$ and~$E'\subseteq E$ corresponds to the vertex set and edge set of a clique of size~$k$, 
then there are~$K$~layers in each phase covered by~$V'\cup E'$.
Hence, having~$K$ layers where no vertices from~$C$ have to be exchanged, 
in each phase~$t$ we can exchange all vertices from~$U^t$ to~$U^{t+1}$.
Starting with set~$S_1=U^1\cup V'\cup E'\cup \{c_1\}$ then yields a solution.

\begin{lemma}%
	\label{lem:whardness-forward}
	Let $(G,k)$ be an instance of \prob{Clique} and
	$(\TG,k',\ell)$ be the instance of \prob{Multistage Vertex Cover}
	resulting from \cref{constr:whardness}.
	If $(G,k)$ is a \yes-instance,
	then $(\TG,k',\ell)$ is a \yes-instance.
\end{lemma}
\begin{proof}
	Let $G' = (V',E')$ be the clique of size $k$ in $G$.
	We construct a 
	solution~$\calS=(S^1_1,\ldots,S^1_{2m},S^1_{2m+1}=S^2_{1},\ldots,S^\kappa_{2m+1} = S^{\kappa+1}_1)$ 
	for $(\TG,k',\ell)$ in the following way.
	For each~$t \in \{1,\ldots,\kappa+1\}$ we set 
	$S^t_1 = V' \cup E' \cup U^t \cup \{c_{(t-1)2m+1}\}$,
	which is a vertex cover of size~$k'$ for~$G_{(t-1)2m+1}$.

	Now, for each~$t \in \{1,\ldots,\kappa\}$, we iteratively construct vertex covers for the layers $(t-1)2m+2$ until $t2m$ in the following way.
	Let $T\ceq (t-1)\cdot 2m$.
	Let $i \in \{1,\ldots,2m-1\}$, 
	and assume that the set $S^t_{i}$ is already constructed and is a vertex cover for~$G_{T+i}$ 
	(this is possible due to the definition of~$S^t_1$).
	We distinguish two cases.
	
	\begin{compactdesc}
	\item[\emph{Case 1: $i$ is odd}.]
	We know that~$c_{T+i}\in S^t_i$.
	If $(S^t_{i}\setminus\{c_{T+i}\}) \cup \{ c_{T+i+2} \}$ is a vertex cover for $G_{T+i+1}$, then we set $S^t_{i+1} = (S^t_{i}\setminus\{c_{T+i}\}) \cup \{ c_{T+i+2} \}$. 
	Otherwise we set $S^t_{i+1} = (S^t_{i}\setminus\{c_{T+i}\}) \cup \{ c_{T+i+1} \}$. 
	In both cases $S^t_{i+1}$ is a vertex cover for $G_{T+i+1}$ and
	either $S^t_{i+1} \cap C = \{ c_{T+i+1} \}$ or $S^t_{i+1} \cap C = \{ c_{T+i+2}\}$.
	
	\item[\emph{Case 2: $i$ is even}.]
	We know that $c_{T+i}$ or $c_{T+i+1}$ is in $S^t_{i}$.
	If $c_{T+i} \in S^t_{i}$, 
	then we set $S^t_{i+1} = (S^t_{i}\setminus\{c_{T+i}\}) \cup \{ c_{T+i+1} \}$,
	which is a vertex cover for $G_{T+i+1}$.
	If $c_{T+i+1} \in S^t_{i}$,
	then~$S^t_{i}$ is already a vertex cover for $G_{T+i+1}$ 
	and the vertices in $V' \cup E'$ cover all edges incident with~$c_{T+i}$ in the graph $G_{T+i}$. 
	In this case we say that $G'$ \emph{covers} the layer $T+i$ and
	set $S^t_{i+1} = (S^t_{i} \setminus \{ u^t_j\}) \cup \{ u^{t+1}_j \}$,
	where $u^t_j$ is an arbitrary vertex in $S^t_{i} \cap U^t$.
    \end{compactdesc}
    
	Observe that the clique~$G'$ covers~$K$ even-numbered layers in each phase.
	Hence, 
	we replace, 
	during phase~$t \in \{1,\ldots,\kappa\}$
	(that is,
	from layer $(t-1)2m+1$ to $t2m+1$), 
	the vertices~$U^t$ with the vertices $U^{t+1}$.
	This also implies that the symmetric difference of two consecutive sets in~$\calS$ is exactly~$2=\ell$.
	It follows that~$\calS$ is a solution for~$(\TG,k',\ell)$.
\lqed
\end{proof}

\subsubsection{Backward Direction}
\label{sssec:bd}

\newcommand{\va}{\ensuremath{\alpha}}
\newcommand{\vb}{\ensuremath{\beta}}

In this section we prove the backward direction for the proof of~\cref{prop:whardness}.
We first show that if an instance of \msvc{} computed by \cref{constr:whardness} is a \yes-instance, 
then it is safe to assume that 
two vertices are neither deleted from nor added to a vertex cover in a consecutive step 
(we refer to these solutions as \emph{smooth}, 
see~\cref{def:smooth}).
Moreover, 
a vertex from the vertex set~$C$ is only exchanged with another vertex from~$C$ and, 
at any time, 
there is exactly one vertex from~$C$ contained in the solution 
(similarly to the constructed solution in \cref{lem:whardness-forward}).
We call these (smooth) solutions \emph{one-centered} (\cref{def:onecentered}).
We then prove that there must be a phase~$t$ for any one-centered solution where at least~$\binom{k}{2}$ times a vertex from ``past'' sets~$U_{t'}$, $t'\leq t$ is deleted.
This at hand, 
we prove that such a phase witnesses a clique of size~$k$.

The fact that a solution needs to contain at least one vertex from~$C$ at any time immediately follows from the fact 
that there is either an edge between two vertices in~$C$ or there is a vertex in~$C$ which is the center of a star with~$k'+1$ leaves.
\begin{observation}
 \label{obs:always1}
 Let~$(\TG,k',\ell)$ from~\cref{constr:whardness} be a \yes-instance.
 Then for each solution~$(S_1,\ldots,S_\tau)$ it holds true that~$|S_i\cap C|\geq 1$ for all~$i\in\set{\tau(\TG)}$.
\end{observation}

\noindent
In the remainder of this section 
we denote the vertices which are removed
from the set~$S_{i-1}$ 
and added to the next set~$S_{i}$ in a solution~$\calS = (\dots, S_{i-1},S_i,\dots)$ 
by
\[ \sdt{S_{i-1}}{S_i}\ceq (S_{i-1}\setminus S_i, S_i\setminus S_{i-1}) .\]
If~$S_{i-1}\setminus S_i$ or~$S_i\setminus S_{i-1}$ 
have size one, then we will omit
the brackets of the singleton.

\begin{definition}
 \label{def:smooth}
 A solution $\calS = (S_1,\ldots,S_\tau)$ for $(\TG,k',\ell)$ from \cref{constr:whardness} is \emph{\smooth}
 if 
  for all $i \in \{2,\ldots,\tau\}$ we have
  $|S_{i-1}\setminus S_{i}|\leq 1$ and $|S_{i-1}\setminus S_{i}|\leq 1$.
\end{definition}

\noindent
In fact,
if there is a solution,
then there is also a smooth solution.

\begin{observation}
 \label{obs:smooth}
 Let~$(\TG,k',\ell)$ from~\cref{constr:whardness} be a \yes-instance.
 Then there is a \smooth{} solution~$(S_1,\ldots,S_\tau)$.
\end{observation}

\begin{proof}
 By \cref{obs:tau2oneedgelayer}, we know that there is a solution~$\calS=(S_1,\ldots,S_\tau)$ such that $|S_1| = k'$ and $k'-1 \leq |S_i| \leq k'$ for all $i \in \{1,\ldots,\tau\}$.
 Hence, for all~$i \in \{2,\ldots,\tau\}$ it holds true that~$\big||S_i|-|S_{i-1}|\big|\leq 1$.
 It follows that $|S_{i-1}\setminus S_{i}|\leq 1$ and $|S_{i-1}\setminus S_{i}|\leq 1$,
 and thus,
 $\calS$~is a \smooth{} solution.
\lqed
\end{proof}

\noindent
Our next goal is to prove the existence of the following type of solutions.

\begin{definition}
 \label{def:onecentered}
 A \smooth{} solution $\calS = (S_1,\ldots,S_\tau)$ for $(\TG,k',\ell)$ from \cref{constr:whardness} is \emph{one-centered}
 if 
 \begin{compactenum}[(i)]
  \item for all $i \in \{1,\ldots,\tau\}$ it holds true that~$|S_i \cap C| = 1$, and
  \item for all $i \in \{2,\ldots,\tau\}$ and $\sdt{S_{i-1}}{S_{i}} = (\va,\vb)$  it holds true that~$\va \in C \iff \vb \in C$.
 \end{compactenum}
\end{definition}

\noindent
We now show that if the output instance of \cref{constr:whardness} is a \yes-instance,
then there is a solution where~$c_1 \in C$ is the only vertex from~$C$ in the first set of the solution.
\begin{lemma}%
 \label{lem:s1only1}
 Let~$(\TG,k',\ell)$ from~\cref{constr:whardness} be a \yes-instance.
 Then there is a \smooth{} solution~$(S_1,\ldots,S_\tau)$ for~$(\TG,k',\ell)$ such that $S_1\cap C=\{c_1\}$.
\end{lemma}
\begin{proof}
 Suppose towards a contradiction that such a \smooth{} solution does not exist.
 That is, in every \smooth{} solution
 the first vertex cover $S_1$ contains at least two vertices from~$C$ 
 (due to~\cref{obs:always1}, 
 $S_1$~must contain at least one).
 Let $\solset$ be the set of \smooth{} solutions with~$|S_1\cap C|$ being minimal, 
 where $S_1$ is the first vertex cover.
 Let~$\mathcal{S}=(S_1,\ldots,S_\tau)\in\solset$ be a \smooth{} solution such that 
 the value~$i\ceq \min\{j\in\set{\tau}\mid c_j\in S_1\setminus\{c_1\}\}$
 is maximal.
 Let~$\mathcal{S'}=(S_1',\ldots,S_\tau')$ be initially~$\calS$.
 
 Suppose there is a~$j\in\set{i-1}$ such that~$\sdt{S_{j}}{S_{j+1}}=(c_i,\va)$. 
 Let~$j'\ceq \min\{j\in\set{i-1} \mid \sdt{S_{j}}{S_{j+1}}=(c_i,\va)\}$ be the smallest among them.
 Then, set~$S_q'\ceq S_q\setminus \{c_i\}$ for all~$q\in\set{j'-1}$  
 to get a feasible solution (note that~$\sdt{S_{j'-1}'}{S_{j'}'}=(\emptyset,\va)$ is feasible since~$|S_{j'-1}'|\leq k-1$).
 This contradicts the minimality of~$\mathcal{S}$ regarding~$|S_1\cap C|$.
 
 Hence, 
 suppose that there is no such~$j$, 
 that is,
 there is no~$j\in\set{i-1}$ such that~$\sdt{S_{j}}{S_{j+1}}=(c_i,\va)$.
 If~$S_i\setminus\{c_i\}$ is a vertex cover of layer~$G_i$, 
 then setting~$S_q'\ceq S_q\setminus\{c_i\}$, 
 for all~$q\in\set{p}$ with~$p\ceq \max\{p'\in\set{\tau}\mid \forall q\in\set{p}:\: c_i\in S_q\}$, 
 yields a feasible solution.
 This contradicts the minimality of~$\mathcal{S}$ regarding~$|S_1\cap C|$.
 
 Finally, 
 suppose that there is no~$j\in\set{i-1}$ such that~$\sdt{S_{j}}{S_{j+1}}=(c_i,\va)$ 
 (and hence~$c_i\in S_i$)
 and~$S_i\setminus\{c_i\}$ is no vertex cover of layer~$G_i$.
 Let~$\sdt{S_{i-1}}{S_i}=(\va,\vb)$ for some~$\va,\vb$ 
 (each being possibly the empty set).
 Then for all~$q\in\set{i-1}$ do the following
 (we distinguish two cases):
 
 \begin{compactdesc}
    \item[\emph{Case 1: $\vb= c_r$ with $r<i$}.]
    Set $S_q '\ceq S_q\setminus \{c_i\}$ and
    $S_{q'}'\ceq S_{q'}\setminus\{\vb\}$ (i.e.~$\sdt{S_{i-1}'}{S_i'}=(\va,c_i)$) 
    for all~$q'\in\set[i]{p}$ with~$p\ceq \max\{p'\in\set{\tau}\mid \forall p''\in\set[i]{p'}:\: \vb\in S_{p''}\}$.
    This contradicts the minimality of~$\mathcal{S}$ regarding~$|S_1\cap C|$.
    
    \item[\emph{Case 2: $\vb= c_r$ with $r>i$, or $\vb\not\in C$}.]
    Set $S_q'\ceq (S_q\setminus \{c_i\})\cup \{\vb\}$ (note that $S_i'=S_i$ and hence~$\sdt{S_{i-1}'}{S_i'}=(\va,c_i)$).
    Note that if there is a~$p\in\set[2]{i-1}$ with~$\sdt{S_{p-1}}{S_p}=(\vb,x)$ or $\sdt{S_{p-1}}{S_p}=(x,\vb)$, 
    then we get~$\sdt{S_{p-1}'}{S_p'}=(\emptyset,x)$ and  $\sdt{S_{p-1}'}{S_p'}=(x,\emptyset)$, respectively.
    In the case of $\vb= c_r$ with $r>i$,
    this contradicts the fact that~$c_i$ is maximal regarding~$i$.
    In the case of $\vb\not\in C$,
    this contradicts the minimality of~$\mathcal{S}$ regarding~$|S_1\cap C|$.
 \end{compactdesc}
 In every case,
 we obtain a contradiction,
 concluding the proof.
 \lqed
\end{proof}

\noindent
Next we show that there are solutions such that 
whenever we remove a vertex in $C$ from the vertex cover,
then we simultaneously add another vertex from~$C$ to the vertex cover. 
Formally, 
we prove the following.
\begin{lemma}%
	\label{lem:center-swaps}
 Let~$(\TG,k',\ell)$ from~\cref{constr:whardness} be a \yes-instance.
 Then there is a \smooth{} solution~$(S_1,\ldots,S_\tau)$ with
 $S_1\cap C=\{c_1\}$ 
 such that 
 for all~$i\in\set{\tau}$ with $\sdt{S_{i-1}}{S_i}=(\va,c)$ and~$c\in C$ we also have~$\va\in C$.
\end{lemma}
\begin{proof}
    Suppose towards a contradiction the contrary.
    That is, 
    let for every \smooth{} solution $(S_1,\dots,S_\tau)$ exist an~$i\in\set{\tau}$ with $\sdt{S_{i-1}}{S_i}=(\va,c)$ and~$c\in C$ and $\va\not\in C$.
    Let $\solset$ be the non-empty (due to \cref{lem:s1only1}) set of \smooth{} solutions $(S_1,\dots,S_\tau)$ with~$|S_1\cap C|=1$.
    Let~$\solset'\subseteq\solset$ be the set of \smooth{} solutions that maximizes the first index~$i$ with~$\sdt{S_{i-1}}{S_i}=(\va,c_q)$ with~$c_q\in C$ and~$\va\not\in C$.
    Among those solutions, 
    consider $\calS = (S_1,\dots,S_\tau) \in \solset'$ to be the one with~$q$ being maximal.
    Note that due to~\cref{obs:always1},
    we have that~$|S_{i-1}\cap C|\geq 1$.
    Let~$S_j'\ceq S_j$ for all~$j\in\set{\tau}$.

    \begin{compactdesc}
    \item[\emph{Case 1: $i>1$ is odd}.]
    Since~$c_i$ is the center of a star in layer~$i$, 
    $c_i$ has to be in~$S_i$.
    We distinguish three subcases regarding the relation of~$q$ and~$i$,
    that is,
    the cases of~$q$ being smaller,
    equal, 
    or larger than~$i$.
 
        \begin{compactdesc}
            \item[\emph{Case 1.1: $q<i$}.]
            Set~$S_j'= (S_j\setminus\{c_q\})$ 
            (i.e.,
            $\sdt{S_{i-1}'}{S_i'}=(\va,\emptyset)$) 
            for all~$j\in\set[i]{q'}$ with~$q'\ceq \max\{q''\in\set[i]{\tau}\mid \forall j\in\set[i]{q''}:\: c_q\in S_j\}$.
            It follows that~$(S_1',\ldots,S_\tau')$ is again a feasible smooth solution contradicting~$i$ being maximal.

            \item[\emph{Case 1.2: $q=i$}.]
            Then~$c_i\not\in S_{i-1}$,
            and hence~$c_{i-1}\in S_{i-1}$ since the edge~$\{c_{i-1},c_i\}$ must be covered in layer~$G_{i-1}$.
            Set~$S_p'= (S_p\setminus \{c_{i-1}\})\cup \{\va\}$ (i.e.,
            $\sdt{S_{i-1}'}{S_i'}=(c_{i-1},c_q)$) 
            for all~$p\in\{i,\ldots,j\}$,
            where~$j>i$ is minimal such that~$\sdt{S_{j-1}}{S_j}=(c_{i-1},x)$, or~$\tau$ if such a~$j$ does not exist.
            If there is a minimal~$j>i$ such that~$\sdt{S_{j-1}}{S_j}=(c_{i-1},x)$, 
            then set~$S_p'= (S_p\setminus\{\va\})$ 
            (i.e.,
            $\sdt{S_{j-1}'}{S_j'} = (\va,x)$) 
            for all~$p\in\set[j]{q'}$ with~$q'\ceq \max\{q''\in\set[i]{\tau}\mid \forall p\in\set[i]{q''}:\:\va\in S_p\}$.
            Suppose that between~$i$ and~$j$, 
            there are~$j_1$ and~$j_2$ such that~$\sdt{S_{j_1-1}}{S_{j_1}}=(y,\va)$ and $\sdt{S_{j_2-1}}{S_{j_2}}=(\va,y')$.
            Note that~$\sdt{S_{j_1-1}'}{S_{j_1}'}=(y,\emptyset)$ and $\sdt{S_{j_1-1}'}{S_{j_1}'}=(\emptyset,y')$.
            It follows that~$(S_1',\ldots,S_\tau')$ is again a feasible smooth solution,
            contradicting~$i$ being maximal.

            \item[\emph{Case 1.3: $q>i$}.]
            Then~$c_i\in S_{i-1}$.
            Let~$\sdt{S_{q-1}}{S_q}=(\vb,d)$.
            We distinguish into two cases regarding~$d$.
            
            \begin{compactdesc}
                \item[\emph{Case 1.3.1: $d=c_p$ with~$p<q$.}]
                Set~$S_j'= S_j\setminus \{c_q\}$ (i.e.,
                $\sdt{S_{i-1}'}{S_i'}=(\va,\emptyset)$) 
                for all~$j\in\set[i]{q-1}$.
                Moreover,
                set~$S_j'=(S_j\setminus \{d\})\cup\{c_q\}$ (i.e.,
                $\sdt{S_{q-1}'}{S_q'}=(\vb,c_q)$) 
                for all~$j\in\set[q]{q'}$ with~$q'\ceq \max\{q''\in\set[q]{\tau}\mid \forall j\in\set[q]{q''}:\:  d\in S_j\}$.
                \item[\emph{Case 1.3.2: $d\not\in C$ or if~$d=c_p$, then~$p>q$.}]
                Set~$S_j'= (S_j\setminus \{c_q\})\cup \{d\}$ (i.e.,
                $\sdt{S_{i-1}'}{S_i'}=(\va,d)$) 
                for all~$j\in\set[i]{q-1}$.
                Moreover,
                set~$S_j'= S_j \cup \{c_q\}$ (i.e.,
                $\sdt{S_{q-1}'}{S_q'}=(\vb,c_q)$ or~$\sdt{S_{q-1}'}{S_q'}=(\vb,\emptyset)$) for all~$j\in\set[q]{q'}$ with~$q'\ceq \max\{q''\in\set[q]{\tau}\mid \forall j\in\set[q]{q'} :\: c_q\in S_j\}$.
            \end{compactdesc}

            In either case,
            we have that~$(S_1',\ldots,S_\tau')$ is a feasible solution contradicting either~$i$ being maximal~($d\not\in C$, or $d=c_p$ with~$p<q$)~or $q$ being maximal~($d=c_p$ with~$p>q$).
        \end{compactdesc}   

        \item[\emph{Case 2: $i>1$ is even}.]
        Then~$c_{i-1}\in S_{i-1}$ and~$c_q\in\{c_i,c_{i+1}\}$.
        Set~$S_j'\ceq (S_j\setminus\{c_{i-1}\})\cup\{\va\}$ 
        (i.e.,
        $\sdt{S_{i-1}'}{S_i'}=(c_{i-1},c_q)$) 
        for all~$j\in\set[i]{q'}$ with~$q'\ceq \max\{q''\in\set[i]{\tau}\mid \forall j\in\set[i]{q''} :\: c_{i-1}\in S_j\}$.
        Then $(S_1',\ldots,S_\tau')$ is a feasible solution contradicting~$i$ being maximal.
    \lqed
    \end{compactdesc}
\end{proof}
Combining \cref{obs:always1,lem:center-swaps}, 
we can assume that for every given \yes-instance, 
there is a solution which is one-centered. 

\begin{corollary}
	\label{lem:one-centered}
 Let~$(\TG,k',\ell)$ from~\cref{constr:whardness} be a \yes-instance.
 Then there is a solution~$\calS$ which is one-centered.
\end{corollary}

\noindent
In the remainder of this section, for each~$t\in\{1,\ldots,\kappa+1\}$ let the union of all~$U^i$ be denoted by
\[ \smod{U}_t \ceq \textstyle{\bigcup^t_{i=1} U^i} .\]

\noindent
We introduce further notation regarding a one-centered 
solution $\calS \ceq (S_1^1,\ldots,S_{2m+1}^1=S_1^2,\ldots,\dots,S_{1}^\kappa,\ldots,S_{2m+1}^\kappa)$ 
for~$(\TG,k',\ell)$.
Here, $S^t_i$ is the $i$-th set of phase~$t$ and thus the~$(2m(t-1)+i)$-th set of~$\calS$.
  The set 
  \begin{align} Y_{i}^t \ceq \{e_j\in S_i^t\cap E\mid 2j\geq i\} \label{eq:Y}\end{align}
  is the set of vertices $e_j$ from $E$ in $S_i^t$ such that the corresponding layer for $e_j$ in phase $t$ is not before the layer $i$ in phase $t$.
  The set 
  \begin{align} F_{i}^t \ceq \{ j > i \mid \sdt{S^t_{j-1}}{S^t_{j}} = (u,\vb)\text{ with } u \in \smod{U}_t \} \label{eq:F}\end{align}
  is the set of layers from~$\TG$ in phase~$t$ where a vertex from~$\smod{U}_t$ is not carried over to the next layer's vertex cover.
We now show that there is a phase~$t$ where~$|F_1^t| \geq K$.
\begin{lemma}%
	\label{lem:whardness-K-flips}
	Let $\calS = (S_1^1,\ldots,S_{2m+1}^1=S_1^2,\ldots,\dots,S_{1}^\kappa,\ldots,S_{2m+1}^\kappa)$ be a one-centered solution to~$(\TG,k',\ell)$ from~\cref{constr:whardness}.
	Then, there is a $t \in \{1,\dots,\kappa\}$ such that $|F^t_1| \geq K$.
\end{lemma}

\begin{proof}
	Suppose towards a contradiction the contrary,
	that is, 
	that for all~$t \in \{1,\dots,\kappa\}$ it holds true that~$|F^t_1| < K$.
	Then, for each~$i\in\{2,\ldots,\kappa+1\}$, we have that~$|S_1^i\cap \smod{U}_{i-1}|\geq i-1$.
	Since~$\calS$ is a solution, we know that~$U^{\kappa+1}\subseteq S_1^{\kappa+1}$ and hence~$|S_1^{\kappa+1}\cap U^{\kappa+1}|=K$.
	Thus, 
	we have that
	\[|S_1^{\kappa+1}|\geq |S_1^{\kappa+1}\cap U^{\kappa+1}|+|S_1^{\kappa+1}\cap \smod{U}_{\kappa}|\geq K+\kappa-1=2K+k+2> k',\]
	contradicting~$\calS$ being a solution.
\lqed
\end{proof}

\noindent
In the remainder of this section,
the value 
\begin{align} 
    f_i^t \ceq |S^t_i \cap \smod{U}_{\kappa+1}| - K  \label{eq:f}
\end{align}
describes the number of vertices in~$\smod{U}_{\kappa+1}$ which we could remove from $S^t_i$ such that $S^t_i$~is still a vertex cover for~$G_{2m(t-1)+i}$ (the $i$-th layer of phase~$t$).
Observe that $f_i^t \geq 0$ for all~$t \in \{1,\dots,\kappa\}$ and all~$i \in \{1,\dots,2m+1\}$,
because we need in each layer exactly~$K$~vertices from~$\smod{U}_{\kappa+1}$ in the vertex cover.

We now derive an invariant which must be true in each phase. 
\begin{lemma}%
	\label{lem:whardness-invariant}
	Let $\calS = (S_1^1,\ldots,S_{2m+1}^1=S_1^2,\ldots,\dots,S_{1}^\kappa,\ldots,S_{2m+1}^\kappa)$ be a one-centered solution to~$(\TG,k',\ell)$ from~\cref{constr:whardness}.
	Then, 
	for all $t \in \{1,\ldots,\kappa\}$ and all~$i \in \{1,\dots,2m+1\}$, 
	it holds true that~$|F^t_i| - |Y^t_i| \leq f^t_i$.
\end{lemma}

\begin{proof}
    Let~$t \in \set{\kappa}$ be arbitrary but fixed.
	For all~$i \in \{1,\dots,2m+1\}$ let 
	\[ \varepsilon_i \ceq |F^t_i| - |Y^t_i| - f^t_i.\]
	We claim that $\varepsilon_i - \varepsilon_{i-1} \geq 0$ for all~$i \in \{1,\dots,2m+1\}$.
	Since~$\calS$ is one-centered, 
	in \cref{tab:binswap} all relevant tuples for~$\sdt{S^t_{i-1}}{S^t_{i}}$ are shown.
	As each relevant tuple results in~$\varepsilon_i - \varepsilon_{i-1}\in\{0,1,2\}$,
	the claim follows.
	\begin{table}
   \caption{Overview of all tuples of~$\sdt{S^t_{i-1}}{S^t_{i}}$ relevant in the proof of~\cref{lem:whardness-invariant} 
   and their possible values of~$\varepsilon_i - \varepsilon_{i-1}= |F_{i}^t|-|F_{i-1}^t|-(|Y_{i}^t|-|Y_{i-1}^t|)-(f_{i}^t-f_{i-1}^t)$.
   In the tuples,
   $u$, $v$, and~$e$ represent some vertex from~$\smod{U}_{\kappa+1}$, $V$, and~$E$, 
   respectively.
   }
   \centering
	 \begin{tabular}{@{}ll|lll|l@{}}\toprule
     $\sdt{S^t_{i-1}}{S^t_{i}}$       &  & $|F_{i}^t|-|F_{i-1}^t|$ & $-(|Y_{i}^t|-|Y_{i-1}^t|)$ & $-(f_{i}^t-f_{i-1}^t)$ & $\varepsilon_i - \varepsilon_{i-1}  $ \\ \midrule
    $(u,\vb)$ &$\vb\in E$  & $\in\{-1,0\}$                       & $\in\{0,1\}$                       & 1                  & $\in\{0,1,2\}$ \\
    & $\vb\in \smod{U}_{\kappa+1}$  & $\in\{-1,0\}$                       & 1                       & 0                  & $\in\{0,1\}$ \\
    & $\vb\in V$, $\vb=\emptyset$  & $\in\{-1,0\}$                       & 1                       & 1                  & $\in\{1,2\}$ \\
    $(\va,u)$ & $\va\in E$  & 0                       & $\in\{1,2\}$                       & -1                  & $\in\{0,1\}$ \\
    & $\va\in V$, $\va=\emptyset$  & 0                       & 1                       & -1                  & 0 \\
    $(\va,v)$ & $\va\in E$  & 0                       & $\in\{1,2\}$                       & 0                  & $\in\{1,2\}$ \\
    & $\va\in V$, $\va=\emptyset$   & 0                       & 1                       & 0                  & 1 \\
    $(\va,e)$ & $\va\in V$  & 0                       & 1                       & 0                  & 1 \\
    & $\va\in E$, $\va=\emptyset$  & 0                       & $\in\{0,1\}$                       & 0                  & $\in\{0,1\}$ \\
	  \bottomrule
	 \end{tabular}
   \label{tab:binswap}
	\end{table}
	
	We want to prove that~$\varepsilon_i \leq 0$ for all~$i \in \set{2m+1}$.
    So,
    assume towards a contradiction that there is a $j \in \{1,\dots,2m+1\}$ such that $\varepsilon_j > 0$.
    Since $\varepsilon_i - \varepsilon_{i-1} \geq 0$ for all~$i \in \{1,\dots,2m+1\}$, 
    we have 
    that $\varepsilon_{2m+1} > 0$,
    which is equivalent to~$|F^t_{2m+1}| - |Y^t_{2m+1}| > f^t_{2m+1}$.
    By definition, 
    we have that~$|Y^t_{2m+1}|=0$ (see~\eqref{eq:Y}) and~$|F^t_{2m+1}|=0$ (see~\eqref{eq:F}).
    Moreover, 
    since~$\calS$ is a solution and each vertex cover needs at least $K$ vertices from $\smod{U}_\tau$, 
    we have that $f^t_{2m+1}\geq 0$.
    It follows that 
    $ 0 = |F^t_{2m+1}| - |Y^t_{2m+1}| > f^t_{2m+1} \geq 0$,
    yielding a contradiction.
    \lqed
\end{proof}

Next, we prove that in a phase~$t$ with~$|F_1^t|\geq K$, 
there are at most~$k$ vertices from~$V$ contained in the union of the vertex covers of phase~$t$.

\begin{lemma}%
	\label{lem:whardness-noofvert}
	Let $\calS = (S_1^1,\ldots,S_{2m+1}^1=S_1^2,\ldots,\dots,S_{1}^\kappa,\ldots,S_{2m+1}^\kappa)$ be a one-centered solution to~$(\TG,k',\ell)$ from~\cref{constr:whardness},
	and let~$t \in \{1,\ldots,\kappa\}$ be such that~$|F_1^t| \geq K$.
	Then, $|\bigcup_{i=1}^{2m+1} S^t_i \cap V|\leq k$.
\end{lemma}

\begin{proof}
  From \cref{lem:whardness-invariant}, we know that $|Y^t_1| \geq K - f^t_1$.
  Let \[ |Y^t_1| = K - f^t_1 + \lambda \] for some $\lambda \in \mathbb N_0$, 
  and let $\varepsilon_i = |F^t_i| - |Y^t_i| - f^t_i$, for all $i \in \{1,\dots,2m+1\}$.

  We now show that there are at most $\lambda$ layers 
  where we exchange a vertex currently in the vertex cover with a vertex in~$V$.
  Let $i \in \{2,\dots,2m+1\}$ such that $\sdt{S^t_{i-1}}{S^t_i} = (\va,v)$ with~$v \in V$.
  From~\cref{tab:binswap} (recall that one-centered solutions are \smooth{}), 
  we know that~$\varepsilon_i \geq \varepsilon_{i-1}+1$.

  Assume towards a contradiction that there are $\lambda+1$ many of these exchanges.
  Then, there is a~$j \in \{1,\dots,2m+1\}$ such that 
  \begin{align*}
   \varepsilon_j &\geq \varepsilon_1 + \lambda + 1 = |F^t_1| - |Y^t_1| - f^t_1 + \lambda + 1 
    \\ &\geq K - (K - f^t_1 + \lambda) - f^t_1 + \lambda + 1 
     \geq 1
    &\iff&& |F^t_j| - |Y^t_j| > f^t_j. 
  \end{align*}
  This contradicts the invariant of \cref{lem:whardness-invariant}.
  
  In the beginning of phase $t$, we have at most $k-\lambda$ vertices from $V$ in the vertex cover, 
	because 
	\[ |S^t_1 \cap V| \leq K + k - |Y^t_1| - f^t_1 = K + k - (K - f^t_1 + \lambda) - f^t_1 = k - \lambda.\] 
	Since there are at most $\lambda$ many exchanges~$\sdt{S^t_{i-1}}{S^t_{i}} = (\va,v)$ where $v\in V$ and $i \in \{2,\dots,2m+1\}$,
	we know that the vertex set $\bigcup_{i=1}^{2m+1} S^t_i \cap V$ is of size at most $k$.
\lqed
\end{proof}

\noindent
We are set to prove the backward direction of~\cref{prop:whardness}.

\begin{lemma}%
	\label{lem:whardness-backward}
	Let $(G,k)$ be an instance of \prob{Clique} and
	$(\TG,k',\ell)$ be the instance of \prob{Multistage Vertex Cover}
	resulting from \cref{constr:whardness}.
	If $(\TG,k',\ell)$ is a \yes-instance,
	then $(G,k)$ is a \yes-instance.
\end{lemma}

\begin{proof}
    Let $(\TG,k',\ell)$ be a \yes-instance.	
    From \cref{lem:one-centered} it follows that there is a 
    one-centered solution $\calS = (S_1^1,\ldots,S_{2m+1}^1=S_1^2,\ldots,\dots,S_{1}^\kappa,\ldots,S_{2m+1}^\kappa)$ for $(\TG,k',\ell)$.
    By \cref{lem:whardness-K-flips}, 
    there is a $t \in \{1,\dots,\kappa\}$ such that $|F_1^t| \geq K= {k \choose 2}$.
    By \cref{lem:whardness-noofvert}, 
    we know that $|\bigcup_{i=1}^{2m+1} S^t_i \cap V|\leq k$.
    Now we identify the clique of size $k$ in $G$.
    Since $|F^t_1| \geq K$, we know that, by \cref{constr:whardness}, 
    at least $K=\binom{k}{2}$ layers are covered by vertices in~$V\cup E\cup \smod{U}_{\kappa+1}\cup \{c_{2j+1}^t\mid j\in\{1,\ldots,m\}\}$ in phase $t$.
    Note that each of these layers corresponds to an edge $e=\{v,w\}$ in $G$ and 
    that we need in particular the vertices $v$ and $w$ in the vertex cover.
    Since we have at most $k$ vertices in $\bigcup_{i=1}^{2m+1} S^t_i \cap V$, 
    these vertices induce a clique of size $k$ in $G$.
    \lqed
\end{proof}

\subsubsection[Proof of the Proposition]{Proof of \cref{prop:whardness} and Two Corollaries}
\label{sssec:proofwhardness}

We proved the forward and backward direction of \cref{prop:whardness} in
\cref{sssec:fd,sssec:bd}, respectively.
It remains to put everything together.

\ifarxiv{}
\begin{proof}[Proof of \cref{prop:whardness}]
		\else{}
\begin{proof}[of \cref{prop:whardness}]
\fi{}
	Let $(G,k)$ be an instance of \prob{Clique} and
	$(\TG,k',\ell)$ be the instance of~\MSVC{} 
	resulting from \cref{constr:whardness}.
	Observe that \cref{constr:whardness} runs in polynomial time,
	and that each layer of $\TG$ is a forest with $O(k'^2)$ edges.
	We know that if $(G,k)$ is a \yes-instance of~\prob{Clique},
	then $(\TG,k',\ell)$ is a \yes-instance of~\MSVC{} (\cref{lem:whardness-forward}),
	and vice versa (\cref{lem:whardness-backward}).
	Finally, 
	the \W{1}-hardness of~\prob{Clique}~\cite{DowneyF99} regarding~$k$ 
	and 
	the fact that~$k'\in O(k^2)$ then finishes the~proof.
\lqed
\end{proof}

From a motivation point of view,
it is natural to assume that the change over time 
modeled by the temporal graph is rather of \emph{evolutionary} character, 
meaning that the difference of a layer to its predecessor is limited. 
However, \cref{prop:whardness} gives a bound (in terms of the desired vertex cover size in input instance)
on the number of edges of each layer.
Hence, we also have the following \W{1}-hardness.
\begin{corollary}
		\todo{pz: do we want to highlight this further? I think it is strong, but does not fit into the storyline.}
		\MSVC{} parameterized by the maximum number~$\max_{i \in\set{\tau}} |E(G_i)|$ of edges in a layer is \W{1}-hard,
		even if each layer is a forest.
\end{corollary}
Thus, we cannot hope for fixed-parameter tractability of \msvc{} when parameterized for example by the combination of $k$
and the maximum size of symmetric difference between two consecutive layers.

Furthermore, we can turn the instance~$(\TG,k',\ell)$ computed by~\cref{constr:whardness} into an equivalent instance~$(\TG',k'',\ell)$ where each layer is a tree as follows.
 Set~$k''=k'+1$.
 Add a vertex~$x$ to~$\TG$.
 In each layer of~$\TG$, 
 make~$x$ the center of a star with~$k''+1$ (new) leaf vertices and connect~$x$ with exactly one vertex of each connected component.
 Note that in every solution $x$ is contained in a vertex cover for each layer in~$\TG'$.
\begin{corollary}
 \label{rem:whardnesstree}
		\MSVC{} parameterized by $k$ is \W{1}-hard,
		even if each layer is a tree.
\end{corollary}
However, 
in \cref{rem:whardnesstree},
$\max_{i \in \tau} |E(G_i)|$ is unbounded and we cannot hope to strengthen the reduction 
in this sense because if each layer is a tree, 
then we have exactly $|V|-1$ edges in each layer.
This would contradict \cref{prop:xpalgo}.

\section{On Efficient Data Reduction}
\label{sec:dataredu}

In this section, 
we study the possibility of efficient and effective data reduction for \MSVC{} when parameterized by~$k$, $\tau$, and~$k+\tau$, 
that is, 
the possible existence of problem kernels of polynomial size.
We prove that unless $\unlessPK$, 
\MSVC{} admits no problem kernel of size polynomial in~$k$~(\cref{ssec:nopkfork}).
Yet, when combining~$k$ and~$\tau$, 
we prove a problem kernel of size~$O(k^2\tau)$~(\cref{ssec:ktaukernel}).
Moreover, we prove a problem kernel of size~$5\tau$ when each layer consists of only one edge~(\cref{ssec:taulinker}).
Recall that \MSVC{} is para-\NP-hard regarding~$\tau$ even if each layer is a tree. 

\subsection{No Problem Kernel of Size Polynomial in~$k$ for Restricted Input Instances}
\label{ssec:nopkfork}

In this section,
we prove the following.\footnote{A graph is planar if it can be drawn on the plane such that no two edges cross each~other.}

\begin{theorem}
 \label{thm:preprock}
 Unless~$\unlessPK$, \MSVC{} admits no polynomial kernel when parameterized by~$k$, 
 even 
 \begin{compactenum}[(i)]
  \item if each layer consists of one edge and~$\ell=1$, or
  \item if each layer is planar and~$\ell\geq 2k$.
 \end{compactenum}
\end{theorem}

\noindent
Recall that  \MSVC{} parameterized by~$k$ 
is fixed-parameter tractable in case of~(ii) (see~\cref{obs:turedu}),
while we left open whether it also holds true in case~(i).

We prove \cref{thm:preprock} using AND-compositions~\cite{BodlaenderDFH09}.

\begin{definition}
  \label{def:andcomposition}
  An \emph{AND-composition} for a parameterized problem~$L$ is an algorithm that,
  given~$p$ instances~$(x_1,k),\ldots,(x_p,k)$ of~$L$, 
  computes in time polynomial in~$\sum_{i=1}^{p}|x_i|$ an instance~$(y,k')$ of~$L$ such that
  \begin{compactenum}[(i)]
  \item $(y,k')\in L$ if and only if~$(x_i,k)\in L$ for all~$i\in\set{p}$, and
  \item $k'$ is polynomially upper-bounded in $k$.
  \end{compactenum} 
\end{definition}

\noindent
The following is the crucial connection to polynomial kernelization.

\begin{theorem}[Drucker \cite{Drucker15}]
 \label{prop:drucker}
 If a parameterized problem whose unparameterized version is \NP-hard admits an AND-composition, then $\unlessPK$. %
\end{theorem}

\noindent
Note that $\unlessPK$ implies a collapse of the polynomial-time hierarchy to its third level~\cite{Yap83}.

In the proof of \cref{thm:preprock}(i), 
we use an AND-composition.
The idea is to take $p$ instances of \msvc{} on the same vertex set with~$\ell=1$ and identical~$k$,
and stack all these instances one after the another in the time dimension.
Here, we connect the $i$-th instance with $(i+1)$-th instance by just repeating 
the first layer of the $(i+1)$-st instance so often such that there is enough time to
\emph{transfer} from a solution of the $i$-th instance to a solution of the $(i+1)$-th instance
without violating the upper bound on the symmetric difference between two consecutive vertex covers.
Formally, we use the following construction.

\begin{construction}
 \label[construction]{constr:ANDc}
 Let~$(\TG_1,k,\ell),\ldots,(\TG_p,k,\ell)$ be~$p$ instances of \MSVC{} where~$\ell=1$ and 
 each layer of each~$\TG_q=(V,\TE_q,\tau_q)$, $q\in\set{p}$, consists of one edge.
 We construct an instance~$(\TG=(V,\TE,\tau),k,\ell)$ of \MSVC{} as follows.
 Denote by~$(G^i_1,\ldots,G^i_{\tau_i})$ the sequence of layers of~$\TG_i$.
 Initially, let~$\TG$ be the temporal graph with layer sequence~$((G^i_j)_{1\leq j\leq \tau_i})_{1\leq i\leq p}$.
 Next, for each~$i\in\set{p-1}$, insert between~$G^i_{\tau_i}$ and~$G^{i+1}_1$ the sequence~$(H_1^i,H_2^i,\ldots,H_{2k}^i)\ceq (G^i_{\tau_i},G^{i+1}_1,\ldots, G^{i+1}_1)$.
 This finishes the construction.
 Note that~$\tau\ceq 2k(p-1)+\sum_{i=1}^p \tau_i$.
 \cqed
\end{construction}

\noindent
In the next two propositions,
we prove that~\cref{constr:ANDc} forms AND-compositions,
used in the proof of \cref{thm:preprock}(i).
\begin{proposition}%
 \label{prop:knoPKononeedge}
  \MSVC{} where each layer consists of one edge and~$\ell=1$ 
  admits an AND-composition when
  parameterized by~$k$.
\end{proposition}

\begin{proof}
 We AND-compose \MSVC{} where each layer consists of one edge.
 Let~$I_1=(\TG_1=(V,\TE_1,\tau_1),k,\ell),\ldots,I_p=(\TG_p=(V,\TE_p,\tau_p),k,\ell)$ be~$p$ instances of \MSVC{} with~$\ell=1$  where each layer consists of one edge.
 Apply~\cref{constr:ANDc} to obtain instance~$I=(\TG=(V,\TG,\tau),k,\ell)$ of \MSVC{}.
 We claim that~$I$ is a \yes-instance if and only if~$I_i$ is a \yes-instance for all~$i\in\set{p}$.
 
 \RD{}
 If~$I$ is a \yes-instance, then for each~$i\in\set{p}$, the subsequence of the solution restricted to the layers~$(G^i_j)_{1\leq j\leq \tau_i}$ forms a solution to~$I_i$.
 
 \LD{}
 Let~$(S^i_1,\ldots,S^i_{\tau_i})$ be a solution to~$I_i$ for each~$i\in\set{p}$.
 Clearly,~$(S^i_1,\ldots,S^i_{\tau_i})$ forms a solution to the layers~$(G^i_j)_{1\leq j\leq \tau_i}$.
 For~$H_1^i$, let~$T_1^i= S^i_{\tau_i}\setminus\{v\}$ for some~$v$ such that the unique edge of~$H_1^i$ is still covered.
 Next, set~$T_2^i=T_1^i\cup\{w\}$, where~$w\in S^{i+1}_1$ with~$w$ being incident with the unique edge of~$H_2^i$.
 Now, over the next~$2k-2$ layers, transform~$T_2^i$ into~$S^{i+1}_1$ by first removing layer by layer the vertices in~$T_2^i\setminus S^{i+1}_1$ (at most~$k-1$ many vertices), and then layer by layer add the vertices in $S^{i+1}_1\setminus T_2^i$ (again, at most~$k-1$ vertices).
 This forms a solution to~$I$.
\lqed
\end{proof}

\noindent
Turning a set of input instances of 
\MSVC{} with only one layer ($\tau=1$) which additionally is a planar graph
into a sequence gives an AND-composition used in the proof of \cref{thm:preprock}(ii).

\begin{proposition}%
 \label{prop:nopkkelltwok}
 \MSVC{} where each layer is planar and $\ell\geq 2k$ 
 admits an AND-composition when parameterized by~$k$.
\end{proposition}

\begin{proof}
 We AND-compose \MSVC{} with one layer being a planar graph (and~$\ell\geq 2k$) into \MSVC{} with~$\ell\geq 2k$.
 Let~$(G_1,k,\ell'),\ldots,(G_p,k,\ell')$ be~$p$-instances of~\MSVC{} with one layer being a planar graph.
 Construct a temporal graph~$\TG$ with layers~$(G_1,\ldots,G_p)$.
 Set~$\ell=2k$.
 This finishes the construction.
 It is not difficult to see that~$(\TG,k,\ell)$ is a \yes-instance of~\MSVC{} if and only if $(G_i,k)$ is a \yes-instance of~\VC{} for all~$i\in\set{p}$.
\lqed
\end{proof}

\noindent
\cref{prop:knoPKononeedge,prop:nopkkelltwok} at hand,
we are set to prove this section's main result.

\ifarxiv{}
\begin{proof}[Proof of \cref{thm:preprock}]
		\else{}
\begin{proof}[of \cref{thm:preprock}]
\fi{}
 Using Drucker's result~\cite{Drucker15} for AND-compositions, 
 \cref{prop:knoPKononeedge,prop:nopkkelltwok} prove \cref{thm:preprock}(i) and (ii), respectively.
 Recall that 
 \MSVC{} where each layer consists of one edge~(\cref{thm:npahrdcases})
 and \MSVC{} on one layer being a planar graph 
 (basically,
 \vc{} on planar graphs)~\cite{GJS76} are \NP-hard.
\lqed
\end{proof}

\subsection[A cubic problem kernel regarding k and tau]{A Problem Kernel of Size~$O(k^2\tau)$}
\label{ssec:ktaukernel}

\MSVC{} remains \NP-hard for~$\tau=2$, 
even if each layer is a tree (\cref{thm:npahrdcases}).
Moreover, \MSVC{} does not admit a problem kernel of size polynomial in~$k$, 
even if each layer consists of only one edge (\cref{thm:preprock}).
Yet, when combining both parameters we obtain a problem kernel of cubic size.

\begin{theorem}
  \label{thm:PKktau}
 There is an algorithm that maps any instance~$(\TG,k,\ell)$ of~\MSVC{} in $O(|V(\TG)|^2\tau)$~time
 to an instance~$(\TG',k,\ell)$ of~\MSVC{} with at most~$2k^2\tau(\TG)$ vertices and at most~$k^2\tau(\TG)$~temporal edges.
\end{theorem}

\noindent
To prove~\cref{thm:PKktau}, we apply three polynomial-time data reduction rules.
These reduction rules can be understood as temporal variants of the folklore reduction rules for \prob{Vertex Cover}.
Our first reduction rule is immediate.

\begin{rrule}[Isolated vertices]
\label{rr:iso}
 If there is some vertex~$v\in V$ such that~$e\cap v=\emptyset$ 
 for all~$e\in E(\UG)$, then delete~$v$.
\end{rrule}

\noindent
For~\prob{Vertex Cover},
when asking for a vertex cover of size~$q$, 
there is the well-known reduction rule dealing with high-degree vertices:
If there is a vertex~$v$ of degree larger than~$q$, then delete~$v$ and its incident edges and decrease~$q$ by one.
For~\MSVC{} a high-degree vertex can only appear in some layers, and hence deleting this vertex is in general not correct.
However, 
the following is a temporal variant of the high-degree rule
(see \cref{fig:highdeg} for an illustration).

\begin{rrule}[High degree]
 \label{rr:highdeg}
 If there exists a vertex~$v$ such that there is an inclusion-maximal subset~$J\subseteq \set{\tau}$ such that~$\deg_{G_i}(v)>k$ for all~$i\in J$, 
 then add a vertex~$w_v$ to~$V$ and for each~$i\in J$, 
 remove all edges incident to~$v$ in~$G_i$, 
 and add the edge~$\{v,w_v\}$.
\end{rrule}

\begin{figure}[t]
  \centering
  \begin{tikzpicture}
  \tikzstyle{xnode}=[circle,fill,scale=1/2,draw]
  \tikzstyle{xnode2}=[fill=lightgray,scale=2/3,draw]
  \tikzstyle{xedge}=[-]

  \def\xr{0.75}
  \def\yr{0.72}

  \newcommand{\tkznodes}[1]{
    \draw[rounded corners, gray] (0,0) rectangle (3*\xr,2*\yr);
    \node at (-0.3*\xr,1.25*\yr)[label=90:{#1}]{};
    \node (v) at (0.25*\xr,1*\yr)[xnode,label=90:{$u$}]{};
    \node (w) at (2.75*\xr,1*\yr)[xnode,label=90:{$v$}]{};
    \foreach \x in {1,...,8}{
      \node (G\x) at (1.5*\xr,2*\yr-0.2*\x*\yr)[circle,scale=0.1,draw]{};	
    }
  }

  \newcommand{\tkznodesX}[1]{
    \tkznodes{#1};
    \node (wv) at (0.25*\xr,0.25*\yr)[xnode2,label=0:{$w_u$}]{};
      \node (ww) at (2.75*\xr,0.25*\yr)[xnode2,label=180:{$w_v$}]{};
  }
  \begin{scope}
    \begin{scope}
      \tkznodes{$G_1$}
      \foreach \x in {1,...,6}{\draw (v) to (G\x);}
      \foreach \x in {1,2}{\draw (w) to (G\x);}
      \draw[fill=white] (1.5*\xr, 1.1*\yr) ellipse [x radius=0.5*\xr,y radius=0.8*\yr];
    \end{scope}

    \begin{scope}[xshift=4*\xr cm]
      \tkznodes{$G_2$}
      \foreach \x in {1}{\draw (v) to (G\x);}
      \foreach \x in {2,...,7}{\draw (w) to (G\x);}
      \draw[fill=white] (1.5*\xr, 1.1*\yr) ellipse [x radius=0.5*\xr,y radius=0.8*\yr];
    \end{scope}

    \begin{scope}[xshift=8*\xr cm]
      \tkznodes{$G_3$}
      \foreach \x in {3,7}{\draw (v) to (G\x);}
      \foreach \x in {2,5,8}{\draw (w) to (G\x);}
      \draw[fill=white] (1.5*\xr, 1.1*\yr) ellipse [x radius=0.5*\xr,y radius=0.8*\yr];
    \end{scope}

    \begin{scope}[xshift=12*\xr cm]
      \tkznodes{$G_4$}
      \foreach \x in {1,...,6}{\draw (v) to (G\x);}
      \foreach \x in {3,...,8}{\draw (w) to (G\x);}
      \draw[fill=white] (1.5*\xr, 1.1*\yr) ellipse [x radius=0.5*\xr,y radius=0.8*\yr];
    \end{scope}
  \end{scope}

  \begin{scope}[yshift=-3.0*\yr cm]
    \begin{scope}
      \tkznodesX{$G_1'$}
      \draw[xedge] (v) -- (wv);
      \foreach \x in {1,2}{\draw (w) to (G\x);}
      \draw[fill=white] (1.5*\xr, 1.1*\yr) ellipse [x radius=0.5*\xr,y radius=0.8*\yr];
      \node at (1.5*\xr,2.5*\yr)[scale=1.2,rotate=-90]{$\leadsto$};
    \end{scope}

    \begin{scope}[xshift=4*\xr cm]
    \tkznodesX{$G_2'$}
      \foreach \x in {1}{\draw (v) to (G\x);}
      \draw[xedge] (w) -- (ww);
      \draw[fill=white] (1.5*\xr, 1.1*\yr) ellipse [x radius=0.5*\xr,y radius=0.8*\yr];
      \node at (1.5*\xr,2.5*\yr)[scale=1.2,rotate=-90]{$\leadsto$};
    \end{scope}

    \begin{scope}[xshift=8*\xr cm]
      \tkznodesX{$G_3'$}
      \foreach \x in {3,7}{\draw (v) to (G\x);}
      \foreach \x in {2,5,8}{\draw (w) to (G\x);}
      \draw[fill=white] (1.5*\xr, 1.1*\yr) ellipse [x radius=0.5*\xr,y radius=0.8*\yr];
      \node at (1.5*\xr,2.5*\yr)[scale=1.2,rotate=-90]{$\leadsto$};
    \end{scope}

    \begin{scope}[xshift=12*\xr cm]
      \tkznodesX{$G_4'$}
      \draw[xedge] (v) -- (wv);
      \draw[xedge] (w) -- (ww);
      \draw[fill=white] (1.5*\xr, 1.1*\yr) ellipse [x radius=0.5*\xr,y radius=0.8*\yr];
      \node at (1.5*\xr,2.5*\yr)[scale=1.2,rotate=-90]{$\leadsto$};
    \end{scope}
  \end{scope}
  \end{tikzpicture}
  \caption{Illustration of~\cref{rr:highdeg}, 
  exemplified for \emph{two} vertices~$u,v$ and~$k=5$. 
  Each ellipse for a graph~$G_i$ and~$G_i'$,
  respectively,
  represents~$G_i-\{u,v\}$ and~$G_i'-\{u,v,w_u,w_v\}$.
  The vertices~$w_v,w_u$ (gray squares) are introduced by the application of~\cref{rr:highdeg}.
  Note that $u$ ($v$) has a high degree in $G_1$ ($G_2$) and $G_4$.
  }
  \label{fig:highdeg}
\end{figure}
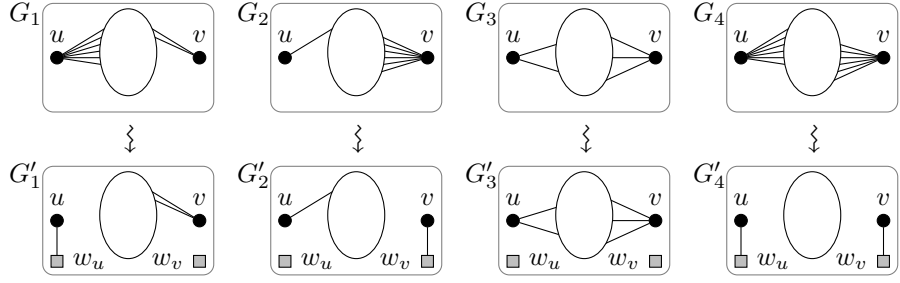

\noindent
We now show how \cref{rr:highdeg} can be applied 
and that it does not turn a \yes-instance 
into a \no-instance or vice versa.

\begin{lemma}%
  \label{lem:highdeg}
 \cref{rr:highdeg} is correct and exhaustively applicable in~$O(|V|^2\tau)$~time.
\end{lemma}

  \begin{proof}
  (\emph{Correctness})
  Let~$I=(\TG,k,\ell)$ be an instance with~$\TG=(G_1,\dots,G_\tau)$, and let~$I'=(\TG',k,\ell)$ be the instance with~$\TG'=(G_1',\dots,G_\tau')$ obtained from~$I$ applying \cref{rr:highdeg} with vertex~$v$ and index set~$J$.
  We prove that~$I$ is a \yes-instance if and only if~$I'$ is a \yes-instance.
  
  \RD{}
  Let~$(S_1,\ldots,S_\tau)$ be a solution to~$I$.
  Observe that for all~$i\in J$, $\deg_{G_i}(v)>k$ and hence~$v\in S_i$.
  It follows that~$(S_1,\ldots,S_\tau)$ is a solution to~$I'$.
  
  \LD{}
  Let~$(S_1',\ldots,S_\tau')$ be a solution to~$I'$.
  Observe that for each~$i\in J$,
  $S_i'\cap\{v,w_v\}\neq \emptyset$.
  Set~$S_i= (S_i'\setminus\{w_v\})\cup \{v\}$ for all~$i\in J$.
  Note that~$S_i$ is a vertex cover for~$G_i$
  since~$v$ covers all its incident edges and~$S_i\setminus\{v\}$ is a vertex cover for~$G_i-\{v\}=G_i'-\{v,w_v\}$.
  For each~$i\in \set{\tau}\setminus J$, 
  set~$S_i= S_i'$ if~$w_v\not\in S_i'$, 
  and~$S_i=(S_i'\setminus\{w_v\})\cup \{v\}$ otherwise. 
  Note that~$S_i$ is a vertex cover of~$G_i=G_i'-\{w_v\}$.
  Finally,
  observe that~$|S_i|\leq |S_i'|$ for all~$i\in\set{\tau}$, and that~$\sydic{S_i}{S_{i+1}}\leq \ell$ for all~$i\in\set{\tau-1}$.
  It follows that~$(S_1,\ldots,S_\tau)$ is a solution to~$I$.
  
  (\emph{Running time})
  For each vertex, we count the number of edges in each layer.
  If there are more than~$k$ edges in one layer, then we remember the index of the layer.
  For each layer, we compute for each vertex the degree and make the modification.
  Once for some~$v$ vertex~$w_v$ is introduced, we add a pointer from~$v$ to~$w_v$, and add the edge~$\{v,w_v\}$ in subsequent layers when needed.
  Hence, in each layer we touch each edge at most twice, yielding~$O(|V(\TG)|^2)$~time per layer.
  \lqed
\end{proof}

\noindent
Similarly as in the reduction rules for~\prob{Vertex Cover},
we now count the number of edges in each layer:
if more than~$k^2$ edges are contained in one layer, 
then no set of~$k$ vertices,
each of degree at most~$k$,
can cover more than~$k^2$ edges.

\begin{rrule}[\no-instance]
 \label{rr:no}
 If neither~\cref{rr:iso} nor~\cref{rr:highdeg} is applicable and there is a layer with more than~$k^2$ edges, then output a trivial \no-instance.
\end{rrule}

\noindent
We are ready to prove that when none of the \cref{rr:iso,rr:highdeg,rr:no} can be applied, then
the instance contains ``few'' vertices and temporal edges.

\begin{lemma}%
  \label{lem:norrsmallsize}
 Let~$(\TG,k,\ell)$ be an instance of~\MSVC{} such that none of \cref{rr:iso,rr:highdeg,rr:no} is applicable.
 Then~$\TG$ consists of at most~$2k^2\tau(\TG)$ vertices and~$k^2\tau(\TG)$ temporal~edges.
\end{lemma}

  \begin{proof}
  Since none of~\cref{rr:iso,rr:highdeg} is applicable,
  for each layer it holds true that there is no isolated vertex and no vertex of degree larger than~$k$.
  Since \cref{rr:no} is not applicable,
  each layer consists of at most~$k^2$ edges. 
  Hence,
  there are at most~$k^2\tau$ temporal edges in~$\TG$.
  Consequently, 
  due to~\cref{rr:iso}, 
  there are at most~$2k^2\tau$ vertices in~$\TG$.
  \lqed
\end{proof}

\noindent
We are ready to prove the main result of this section.

\ifarxiv{}
\begin{proof}[Proof of~\cref{thm:PKktau}]
		\else{}
\begin{proof}[of~\cref{thm:PKktau}]
\fi{}
 Given an instance~$I=(\TG,k,\ell)$ of~\MSVC{},
 apply \cref{rr:iso,rr:highdeg,rr:no} exhaustively in~$O(|V(\TG)|^2\tau(\TG))$ time either to  decide that~$I$ is a trivial \no-instance or
 to obtain an instance~$(\TG',k,\ell)$ equivalent to~$I$.
 Due to~\cref{lem:norrsmallsize}, 
 $\TG'$ consists of at most~$2k^2\tau(\TG)$ vertices and at most~$k^2\tau(\TG)$ temporal edges.
\lqed
\end{proof}

\subsection[A linear problem kernel regarding tau]{A Problem Kernel of Size \(5\tau\)}
\label{ssec:taulinker}

\MSVC{},
even when each layer is a tree,
does not admit a problem kernel of any size in~$\tau$ unless~$\classP=\NP$.
Yet,
when each layer consists of only one edge, 
then each instance of \MSVC{} contains at most~$\tau$ edges and, 
hence, 
at most~$2\tau$ non-isolated vertices.
Thus,
\MSVC{} admits a straight-forward problem kernel of size linear in~$\tau$.

\begin{observation}
 \label[observation]{thm:preproctau}
 Let~$(\TG,k,\ell)$ be an instance of~\MSVC{} where each layer consists of one edge.
 Then we can compute in~$O(|V(\TG)|\cdot \tau)$ time an instance $(\TG',k,\ell)$ of size at most~$5\tau(\TG)$.
\end{observation}

\begin{proof}
 Let~$(\TG,k,\ell)$ be an instance of~\MSVC{} where each layer of~$\TG=(V,\TE,\tau)$ consists of one edge.
 Observe that we can immediately output a trivial \yes-instance
 if~$k\geq \tau$ (\cref{obs:tau2oneedgelayer}) or~$\ell\geq 2$~(\cref{obs:turedu}).
 Hence, assume that~$k\leq\tau-1$ and~$\ell\leq 1$.
 Apply~\cref{rr:iso} exhaustively on~$(\TG,k,\ell)$ to obtain~$(\TG',k,\ell)$.
 Since there are~$\tau$ edges in~$\TG$, there are at most~$2\tau$ vertices in~$\TG'$.
 It follows that the size of~$(\TG',k,\ell)$ is at most~$5\tau$.
\lqed
\end{proof}

\section{Conclusion}
\label{sec:conclusion}

We introduced \msvc{}, 
proved it to be \NP-hard even on very restricted input instances, 
and studied its parameterized complexity regarding the natural parameters~$k$, $\ell$, and~$\tau$ (each given as input). %
A highlight is the $\wone$-hardness described in \cref{ssec:whardness} 
which, because it holds on very restricted instances of \msvc{}, may turn out to be useful to provide $\wone$-hardness results for other problems in the multistage setting. 
We leave open whether \MSVC{} parameterized by~$k$ is fixed-parameter tractable
when each layer consists of only one edge (see \cref{tab:results}).
Moreover, it is open whether \MSVC{} remains \NP-hard on two layers each being a path (that is, strengthening \cref{thm:npahrdcases}(i)).

  \bibliographystyle{plainnat}
\bibliography{msvc}

\end{document}